\documentclass[a4paper,USenglish]{lipics-v2021}

\usepackage{listings}
\usepackage[dvipsnames]{xcolor}
\definecolor{mygreen}{rgb}{0,0.6,0}
\definecolor{mygray}{rgb}{0.5,0.5,0.5}

\lstdefinestyle{imandra}{
    aboveskip=5pt,
	language=caml,
	keywordstyle=\color{blue},
	commentstyle=\color{green},
	emphstyle=\color{mygreen},
    commentstyle=\itshape\color{purple},
	stringstyle=\color{brown},
	basicstyle=\ttfamily \footnotesize,
	breakatwhitespace=false,         
	breaklines=true,                 
	captionpos=b,       
	keepspaces=true,
	showspaces=false,                
	showstringspaces=false,
	showtabs=false,                  
	tabsize=2,
	inputencoding=utf8,
	extendedchars=true,
	literate={φ}{{$\phi$}}1,
	numbers=none,
	stepnumber=1,
	backgroundcolor=\color{white} 
}

\usepackage{caption}
\usepackage{multicol}
\usepackage{algorithm}
\usepackage[noend]{algpseudocode}
\usepackage{framed}
\usepackage{amsmath, amssymb, amsthm}   
\usepackage{graphics}
\usepackage{verbatim}
\usepackage{cancel}
\usepackage{xspace}
\usepackage{tikz,ifthen,pgfplots}
\usepackage{mathrsfs}  
\usepackage[bottom]{footmisc}
\usepackage{enumerate}
\usepackage{booktabs}

\usetikzlibrary{arrows,trees,backgrounds,automata,shapes,decorations,plotmarks,fit,calc,positioning,shadows,chains}
\usetikzlibrary{shapes.symbols}
\tikzstyle{every pin edge}=[<-,shorten <=1pt]
\tikzstyle{neuron}=[circle,fill=black!25,minimum size=17pt,inner sep=0pt]
\tikzstyle{input neuron}=[neuron, fill=green!50]
\tikzstyle{output neuron}=[neuron, fill=red!50]
\tikzstyle{hidden neuron}=[neuron, fill=blue!50]
\tikzstyle{annot} = [text width=6em, text centered]
\tikzstyle{nnedge} = [-{stealth},shorten >=0.1cm, shorten <=0.05cm,line width=0.8pt,black]
\tikzstyle{proofNode} = [rounded rectangle, fill=red!30]
\tikzstyle{proofLeaf} = [rounded rectangle, fill=orange!50]
\tikzstyle{proofEdge} = [-{stealth},shorten >=0.1cm, shorten <=0.05cm,line width=0.8pt,black]
\usetikzlibrary{calc}

\newcommand{\relu}{\text{ReLU}\xspace}
\newcommand{\sat}{\texttt{SAT}\xspace}
\newcommand{\unsat}{\texttt{UNSAT}\xspace}

\newcommand{\rn}[1]{\mathbb{R}^{#1}}
\newcommand{\nn}{\mathcal{N}}

\usepackage{letltxmacro}
\newcommand*{\SavedLstInline}{}
\LetLtxMacro\SavedLstInline\lstinline
\DeclareRobustCommand*{\lstinline}{%
	\ifmmode
	\let\SavedBGroup\bgroup
	\def\bgroup{%
		\let\bgroup\SavedBGroup
		\hbox\bgroup
	}%
	\fi
	\SavedLstInline
}
\newcommand{\imlcode}[1]{\lstinline[style=imandra,morekeywords={Leaf,Node,Split,Relu,SingleVar},moreemph={list,int}]{#1}}

\newcolumntype{S}{>{\scriptsize}p}

\newcommand{\R}{\mathbb{R}}

\newcommand{\dnnvquery}[1]{\langle #1 \rangle}
\newcommand{\satpred}[1]{\texttt{is\textunderscore solution}(#1)}
\newcommand{\satpredsys}[1]{\texttt{eval\textunderscore system}(#1)}
\newcommand{\mkcert}[1]{\texttt{mk\textunderscore certificate}(#1)}
\newcommand{\node}{\mathcal{T}\xspace}
\newcommand{\vars}{\x}

\newcommand{\x}{\mathbf{x}}

\newcommand{\inp}{\mathbf{I}}
\newcommand{\out}{\mathbf{O}}
\renewcommand{\u}{\mathbf{u}}
\renewcommand{\l}{\mathbf{l}}
\newcommand{\unode}[1]{\u^{#1}}
\newcommand{\lnode}[1]{\l^{#1}}

\newcommand{\sv}{\mathbf{s}}
\renewcommand{\c}{\mathbf{c}}

\newcommand{\checknode}{\texttt{check\textunderscore tree}}
\newcommand{\updatebounds}{\texttt{update\textunderscore bounds}}

\newcommand{\sys}{\mathcal{S}}
\newcommand{\matrixtopoly}[1]{\sys^{#1}}
\newcommand{\splt}{Split}

\theoremstyle{definition}

\newcounter{rowcntr}[table]
\renewcommand{\therowcntr}{1.\arabic{rowcntr}}

\AtBeginEnvironment{tabular}{\setcounter{rowcntr}{0}}
\newcolumntype{N}{>{\refstepcounter{rowcntr}\therowcntr}c}

\usepackage{todonotes}

\title{A Certified Proof Checker for Deep Neural Network Verification in Imandra}

\author{Remi Desmartin\footnote{Both authors contributed equally} }{Heriot-Watt University, Edinburgh, United Kingdom}{}{}{}
\author{Omri Isac$^1$}{The Hebrew University of Jerusalem, Israel}{}{}{}
\author{Grant Passmore}{Imandra Inc.}{}{}{}
\author{Ekaterina Komendantskaya}{Southampton and Heriot-Watt Universities, United Kingdom}{}{}{}
\author{Kathrin Stark}{Heriot-Watt University, Edinburgh, United Kingdom}{}{}{}
\author{Guy Katz}{The Hebrew University of Jerusalem, Israel}{}{}{}
\Copyright{Remi Desmartin, Omri Isac,  Ekaterina Komendantskaya, Kathrin Stark, Grant Passmore and Guy Katz}


\relatedversion{}
\supplement{}
\funding{The work of Komendantskaya was partially supported by the EPSRC grant AISEC: AI Secure and Explainable by Construction (EP/T026960/1) and ARIA grant ``Safeguarded AI". Desmartin acknowledges Imandra sponsorship for his PhD studies. The work of Isac and
Katz was partially funded by the European Union (ERC, VeriDeL, 101112713). Views and opinions expressed are however those of the author(s) only and do not necessarily reflect those of the European Union or the European Research Council Executive Agency. Neither the European Union nor the granting authority can be held responsible for them. }
\acknowledgements{}
\nolinenumbers
\hideLIPIcs

\begin{document}
\authorrunning{R. Desmartin, O. Isac  et al.}
\maketitle
\begin{abstract}
 Recent advances in the verification of deep neural networks (DNNs)
 have opened the way for a broader usage of DNN verification technology in many application areas, including safety-critical ones. However, DNN verifiers are themselves complex programs that have been shown to be susceptible to errors and numerical imprecision; this, in turn, has raised the question of trust in DNN verifiers. One prominent attempt to address this issue is enhancing DNN verifiers with the capability of producing certificates of their results that are subject to independent algorithmic checking.
 While formulations of Marabou certificate checking already exist on top of the state-of-the-art  DNN verifier Marabou, 
 %the native implementation of the proof checking algorithm for Marabou was done 
 they are implemented in C++, and that code itself raises the question of trust (e.g., in the precision of floating point calculations or guarantees for implementation soundness).
 % \omri{can we make a stronger statement here? something like "we present a \emph{formally verified} alternative ... that has a formal guarantee for its soundness, including ..." ?} 
 Here, we present an alternative implementation of the Marabou certificate checking in Imandra -- an industrial functional programming language and an interactive theorem prover (ITP) -- that allows us to obtain full proof of certificate correctness.
 %DNN proof certification with formal guarantees. 
 The significance of the result is two-fold.
 Firstly, it gives stronger independent guarantees for Marabou proofs. Secondly, it opens the way for the wider adoption of DNN verifiers in interactive theorem proving in the same way as many ITPs  already incorporate SMT solvers. 
 %and, ultimately, wider adoption of DNN verifiers in interactive theorem proving.
 %Moreover, interactive Theorem Provers (ITPs) should be able to incorporate input of DNN verifiers in the same way as many ITPs  incorporate SMT solvers. 
 %, including proofs of mathematical results underlying the algorithm, mainly Farkas' lemma tailored for handling the non-linear constraints present in DNN verification.
\end{abstract}

\section{Introduction}
\label{sec:intro}

As part of the AI revolution in computer science, deep neural networks (DNNs) are becoming the state-of-the-art solution to many computational problems, including 
in safety-critical applications, e.g., in medicine~\cite{suzukiOverviewDeepLearning2017}, aviation~\cite{ElElIsDuMeGaPoGuBoCoKa24} and autonomous vehicles~\cite{BoDeDwFiFlGoJaMoMuZhZhZhZi16}.
DNNs are complex functions, with parameters \emph{optimised} (or \emph{trained}) to fit a large set of input-output examples.
Therefore, DNNs are intrinsically opaque for humans and are known for their vulnerability to errors in the presence of distribution shifts~\cite{szegedyIntriguingPropertiesNeural2013}. 
This raises the question of the reliability of DNNs in safety-critical domains and calls for tools that guarantee their safety~\cite{PLNNV25}.

 One family of such tools 
 are \emph{DNN verifiers}, which allow to either mathematically prove that a DNN complies with certain properties or provide a counterexample for a violation~\cite{BrMuBaJoLi23}.
They are typically based on SMT-solving with bound tightening (Marabou~\cite{WuIsZeTaDaKoReAmJuBaHuLaWuZhKoKaBa24}), interval bound propagation ($\alpha\beta$-CROWN~\cite{wangBetaCROWNEfficientBound2021}), or abstract interpretation (AI2~\cite{GeMiDrTsCgVe18}). 
But implementing these verifiers is only a first step to full safety guarantees.
Even complete DNN verifiers are prone to implementation bugs
and numerical imprecision that, in turn, might compromise their soundness and can be maliciously exploited, as was shown in e.g.~\cite{PLNNV25,JiRi21,Zombori2021FoolingAC}.

One could consider verifying DNN verifiers directly. However, as Wu et al.~exemplify, a mature DNN verifier is a complex multi-platform library; its direct verification would be close to infeasible~\cite{WuIsZeTaDaKoReAmJuBaHuLaWuZhKoKaBa24}. A similar problem was faced by the SMT solvers community~\cite{Ne98, BaDeFo15,BaReKrLaNiNoOzPrViViZoTiBa22, DeBj11}. A solution was found: to produce (and then check)  a \emph{proof certificate} for each automatically generated proof, instead of checking the verifier in its entirety~\cite{10.1145/3689624}. The software that checks the proof certificate is called the \emph{proof checker}.  
Ideally,  the proof checker should: a) be significantly simpler than the original verifier and b) yield strong guarantees of code correctness (i.e. be formally verified). 
For Marabou, the first half of this research agenda was accomplished by Isac et al.~\cite{IsBaZhKa22} with
proof checking for Marabou being reduced to the application of
the \emph{Farkas lemma}~\cite{Va96}, a well-known solvability theorem
for a finite system of linear equations, along with a tree structure that reflects Marabou's verification procedure.  
However, the  proof checker in~\cite{IsBaZhKa22}
came with no formal guarantees of its own correct implementation. 
\begin{figure}[t]
\vspace{-0.6cm}
\centering
\begin{tikzpicture}[
   verifier/.style={rectangle, , fill=orange!40, thick, minimum width=1cm, minimum height=1cm, align=center,
  rounded corners=2pt
  },
  input/.style={rectangle,, fill=yellow!40, thick, minimum width=1cm, minimum height=1cm, align=center,
  rounded corners=2pt
  },
  trusted/.style={rectangle, fill=green!40, thick, minimum width=1cm, minimum height=1cm, align=center,
  rounded corners=2pt
  }]

  \node[input](DNN){DNN};
    \node[input, below = 0.5cm of DNN](spec){Property \\ Spec.};
    \node[verifier, yshift = -0.75cm, right = 1.1cm of DNN](verifier){Certificate-producing \\ DNN Verifier};

    \node[trusted, right = 3.1cm of verifier, yshift = 0.75cm](SAT){Check  in DNN};

     \node[trusted, below= 0.5cm of SAT,](UNSAT){
     Obtain full proofs \\ from certificates  (in ITP)
   %  Check certificates with \\ verified  implementation \\ in Imandra
   };

   \draw[->] (DNN) -- (verifier);
   \draw[->] (spec) -- (verifier);
   \draw[->] (verifier) -- (SAT)node[midway, above,yshift=5pt] {Counter-Example};
   \draw[->, ForestGreen, line width=1pt] (verifier) -- (UNSAT) node[midway, below,yshift=-5pt, black] {Compliance};
\end{tikzpicture}
\caption{Model of trust: a neural network (DNN) and a property specification verified by a certificate producing DNN Verifier (e.g. Marabou). The SAT/UNSAT certificate is then checked by a trusted program in both cases. This paper addresses verification of the \unsat certificate, using the ITP Imandra.}
  \label{fig:trust}
\end{figure}

In this paper, we address this problem by
 deploying the well-established principles of \emph{proof-carrying code}~\cite{Necula97,komendantskaya2025proofcarryingneurosymboliccode} and \emph{self-certifying code}~\cite{10.1145/3689624}. Namely, we implement formal proof production from Marabou certificates in
the interactive theorem prover Imandra~\cite{passmoreImandraAutomatedReasoning2020}.
Imandra can implement -- and prove properties of -- programs written in a subset of OCaml, and thus avoids a potential discrepancy between implementation and formalisation. In Imandra, we obtain a certificate checker implementation that can be safely executed \emph{and} proven correct in the same language.
Imandra supports arbitrary precision real arithmetic, and thus does not suffer from floating point imprecision that haunts DNN verifiers. Finally, it strikes a good balance between interactive and automated proving: i.e., it features strong automation while admitting user interaction with the prover via tactics. 
We exhibit the role of the Imandra checker within the verification process in Figure~\ref{fig:trust}.

\noindent
\textbf{Contributions.} 
Reconstruction of full correctness proofs from Marabou certificates rests on two major results:
\begin{itemize}
\item  Firstly, the  Imandra implementation of an algorithm that reconstructs a proof of validity of Marabou certificates. The latter are given by \emph{Marabou proof trees} with a DNN verification query at the root, nodes corresponding to partitions of the search space, and Farkas vectors at the leaves. We get for free that our implementation of the algorithm is free from floating point instability, unlike the C++ version in~\cite{IsBaZhKa22}. 

\item Secondly, a formal proof of soundness of the algorithm proven in Imandra, that guarantees that  if this algorithm returns  \unsat then indeed the given set of polynomial constraints at the root of the tree has no solution. This proof relies on a new formal proof of the \emph{DNN variant} of the famous Farkas lemma in Imandra. 
\end{itemize}

 Finally, we demonstrate that our verified certificate checker can be executed in practice. We analyze the code complexity, run it on two verification benchmarks, and evaluate its performance speed against the original certificate checker and verifier~\cite{IsBaZhKa22}. Our results suggest an expected trade-off between the reliability and scalability of the proof checking process, with Imandra code taking about 4.56-4.76 times as long as original verification and checking time, on average.

The paper proceeds as follows. We start with explaining the essence of Marabou certificate production in Section~\ref{sec:background} and introducing the Imandra proof of the Farkas lemma in Section~\ref{sec:check}. Section~\ref{sec:imandra_checker} introduces the new certificate checking algorithm in Imandra, and Section~\ref{sec:checkerSoundness} proves its soundness. Finally, Section~\ref{sec:eval} is devoted to practical evaluation of the verified checker, and Section~\ref{sec:relatedwork} concludes the paper.

\section{Background}
\label{sec:background}

\subsection{Deep Neural Networks as Graphs}
\label{sec:DNN}
A DNN is usually defined as a function $\nn:\rn{k_I}\rightarrow\rn{k_O}$, with inputs $\x_\inp = x_1, \ldots, x_{k_I}$ and outputs $\x_\out = y_1, \ldots, y_{k_O}$.
We represent a DNN by a directed, acyclic, connected and weighted graph $(V,E)$, where $V$ is a finite set of vertices $\{v_1, \ldots , v_k\}$ and $E$ is a set of edges of the form $\{(v_i,v_j)\}$, equipped with a \emph{weight function} $w: E \to \mathbb{R}$, and an \emph{activation function} $a_v:  \R \to \R$ for each node $v$.
A vertex $v_i$ is called an \emph{input vertex} if  there is no $v_j$ s.t. $(v_j, v_i) \in E$ and
 an \emph{output vertex} if  there is no $v_j$ s.t. $ (v_i, v_j) \in E$. The following picture shows an example of a DNN represented as a graph:
 
\vspace{-0.2cm}
   \begin{center}
  \def\layersep{1.5cm}
  \scalebox{0.9}{
      \begin{tikzpicture}[shorten >=1pt,->,draw=black!50, node distance=\layersep,font=\footnotesize]
        
        \node[input neuron] (I-1) at (0,-1) {$x_1$};
        \node[input neuron] (I-2) at (0,-2.5) {$x_2$};
    
        \node[hidden neuron] (H-1) at (\layersep, -1) {$v_1$};
        \node[hidden neuron] (H-2) at (\layersep,-2.5) {$v_2$};
        \node[output neuron] (O-1) at (2*\layersep,-1.75) {$y$};
        
        % Connect every node in the hidden layer with the output layer 
        \draw[nnedge] (I-1) --node[above,pos=0.4] {$2$} (H-1);
        \draw[nnedge] (I-2) --node[above,pos=0.4] {$1$} (H-2);
        \draw[nnedge] (I-1) --node[above,pos=0.25] {$-1$} (H-2);
        \draw[nnedge] (I-2) --node[above,pos=0.25] {$1$} (H-1);
        
        \draw[nnedge] (H-1) --node[above] {$-1$} (O-1);
        \draw[nnedge] (H-2) --node[below] {$2$} (O-1);

        % Activation
        \node[above=0.01cm of H-1] (b1) {$\relu$};
        \node[above=0.01cm of H-2] (b1) {$\relu$};
        \node[above=0.01cm of O-1] (b1) {$id$};
      \end{tikzpicture}
  }
\end{center}

We assume that all input vertices are assigned a value,  using $f(v_i) = r \in \mathbb{R}$.
For all other vertices, we define the \emph{weighted sum} $(z)$ and \emph{neuron} $(f)$ functions as follows:
\begin{equation}
\vspace{-0.2cm}
\label{eq:DNNdef}
z(v_j) =  \underset{(v_i, v_j) \in E}{ \sum }    w(v_i, v_j) \cdot f(v_i)
\\
f(v_j) = a_{v_j} ( z(v_j))
\end{equation}

\noindent Note that in this paper, we will use only two activation functions: 
\textit{rectified linear unit} (\relu{}), defined as  $\relu(x) = \max(x,0) $ and the
 identity function $id$. However, in principle, this work can be extended to support DNNs with any piecewise linear activation function.

To cohere with the existing literature, we will use $x$ and $y$ (with indices) to refer to input and output vertices and $\x_\inp$ and $\x_\out$ to lists of input and output vertices. For brevity, we will call them simply \emph{inputs} and \emph{outputs}.

\subsection{DNN Verification Queries}
\label{sec:DNNverification}

Given a variable vector $\x$ of size $k_I$, and $\l, \u \in \rn{k_I}$, a \emph{bound property} $P(\x)$ is defined by inequalities $l_j\leq x_j \leq u_j$ for each element $x_j$ of $\x$.
Given a DNN $\nn$,
an \emph{input bound property} $P(\x_\inp)$ 
and a \emph{linear arithmetic} output property $Q(\x_\out)$,
such that $ P(\x_\inp) \wedge Q(\nn(\x_\inp))$.
 If such $\x_\inp$ 
 exists, the verification problem $\langle \nn, P, Q \rangle$ is
\textit{satisfiable} (\sat); otherwise, it is
\textit{unsatisfiable} (\unsat). 
Typically, $ P(\x_\inp) \wedge Q(\nn(\x_\inp)) $ represent an erroneous behavior; thus
an input $\x_\inp$ and output $\nn(\x_\inp)$ that satisfy the query serve as a counterexample and 
\unsat{} implies safe behavior. Although our approach supports any linear property for the output, we focus on bound properties for simplicity.
This definition can be extended to more sophisticated properties of  DNNs, that express constraints on the a DNN's internal variables~\cite{GoLuMaPuXiYu23}.

Recall that all DNNs we consider 1.) come with affine (i.e., linear polynomial) weight functions and 2.) piecewise linear activation functions. 
As a consequence, any DNN verification problem can be reduced to piecewise linear constraints that represent the activation functions, accompanied with a Linear Programming (LP)~\cite{Da63} instance, which represents the affine functions and the input/output properties. 
One  prominent approach for solving the DNN  verification constraints is using algorithms for solving LP, coupled with a \emph{case-splitting} approach for handling
the piecewise linear constraints~\cite{BaIoLaVyNoCr16, KaBaDiJuKo21}.

One example of such a DNN verification algorithm is the Reluplex~\cite{KaBaDiJuKo21} algorithm, which extends the Simplex algorithm~\cite{Da63,DuDe06} for solving LP problems. 
Based on a DNN verification problem $\langle \nn, P, Q \rangle$, the algorithm initiates:
\begin{itemize}
\item a variable vector $\x$ containing variables $\x_\inp$ and $\x_\out$, and fresh variables $x_{z(v_i)}$  and $x_{f(v_i)}$ whose values represent weighted sum and neuron functions $z(v_i)$ and $f(v_i)$ for all $v_i \in V$ which are neither an input nor an output. For each node $v_i\in V$ with the \relu{} activation function, the algorithm initiates an additional fresh \emph{auxiliary variable} $x_{aux_i}$, to represent the non-negative difference $f(v_i)-z(v_i)$.

\item two \emph{bound vectors} $\u,\l$, giving upper and lower bounds to each element in $\x$. 
The values for $\l$ and $\u$ are generated as follows: values of $\x_\inp$ and $\x_\out$ are given directly by $P(\x_\inp)$ and $Q(\x_\out)$. The remaining values of $\l$ and $\u$ are computed by propagating forward the lower and upper bounds of $\x_\inp$ through $\nn$ using equations~(\ref{eq:DNNdef}) and~(\ref{eq:tableau_lpe}). For variables $x_{f(v_i)}, x_{z(v_i)}, x_{aux_i}$, we denote their lower and upper bounds $\l_{f(v_i)}, \l_{z(v_i)}, \l_{aux_i}$ and $\u_{f(v_i)}, \u_{z(v_i)}, \u_{aux_i}$, respectively.

\item a matrix $A$, called \emph{tableau}, which represents the equations given in~(\ref{eq:DNNdef}), with additional \emph{auxiliary} equations. It is computed by defining a system of 
equations of the form:
\begin{equation}
\vspace{-0.2cm}
\label{eq:tableau_lpe} \underset{(v_i, v_j) \in E}{ \sum }    w(v_i, v_j) \cdot x_{f(v_i)} - x_{z(v_i)} = 0 
\\
x_{f(v_i)} - x_{z(v_j)}- x_{aux_i} = 0 
\end{equation} 
one for each non-input node $v_i$.  
In  $A$, we record these coefficients as follows. Let
$A \in \R^{m \times n}$, where $n$ is the size of $\x$ and $m$ is the number of 
equations we generated as in~(\ref{eq:tableau_lpe}).  Each entry $(k,j)$ in $A$ contains the coefficient of the $j^{th}$ variable in $\x$ in the $k^{th}$ equation.
We then obtain $A \cdot \x=\mathbf{0}$, where $\cdot$ is the dot product, and $\mathbf{0}$ is the vector of $m$ zeros.

\item a set  $C$ of \emph{ReLU constraints} for variables in $\x$, given by $x_{f(v_i)} = ReLU (x_{z(v_i)})$ for all non-input nodes $v_i$, such that $f(v_i) = ReLU(z(v_i))$ (as per equation~(\ref{eq:tableau_lpe})). We call the variables $x_{f(v_i)}$, $x_{z(v_i)}$ and $x_{aux_i}$ the \emph{participating variables} of a constraint.
We can join all $c \in C$  in a conjunction: $\bigwedge c$, and  denote the resulting formula by $\mathcal{C}(\x)$. Intuitively, $\mathcal{C}(\x)$ holds if the variables in $\x$ satisfy the constraints in $C$.

\end{itemize}

\noindent    The tuple $\dnnvquery{A, \u, \l, C}$ is a \emph{DNN verification query}, and the tuple $\dnnvquery{A, \u, \l}$ is a \emph{linear query}. We use the notation $\x[\x_i/a]$ for substituting the $i^{th}$ element in a vector $\x$ by the value $a$.

\begin{example}[DNN Verification Query]
\label{ex:query}
Consider the DNN in Section~\ref{sec:DNN},
the input bound property $P$ that holds if and only if $(x_1, x_2)\in [0,1]^2$ and the output bound property $Q$ that holds if and only if $y\in[4, 5] $. We first obtain $\x$:

\begin{tabular} {c}
        $\x =
            \begin{bmatrix}
                x_1 & x_2 & x_{z(v_1)}  & x_{z(v_2)} & x_{f(v_1)}& x_{f(v_2)} &  x_{aux_1} &  x_{aux_2} & y 
            \end{bmatrix}^\intercal$ \\
\end{tabular}

\noindent Next, $\l$ and $\u$ are computed as follows. Firstly, $P$ and $Q$ already give the bounds for $x_1$, $x_2$, $y$. For the rest, we propagate the input bounds. We get:
\begin{equation}
\begin{split}
     &0 \leq x_1, x_2 \leq 1,\:\:  0\leq x_{z(v_1)}, x_{f(v_1)} \leq 3, \:\:
     -1 \leq x_{z(v_2)} \leq 1, \:\: 0 \leq x_{f(v_2)} \leq 1, \\ 
      &0\leq x_{aux_1} \leq 3, \:\:
      0 \leq  x_{aux_2} \leq 2, \: \: 
      4 \leq y \leq 5
\end{split}
\end{equation}
It only takes to assemble these values into vectors $\l$ and $\u$: 

\begin{tabular} {c}
         $\u =
            \begin{bmatrix}
               1 & \:  1 & \:\:\: \: 3 & \:\: 1 & \:\: 3 & \:\: 1 &  \: 3 & \: 2 & \: 5
            \end{bmatrix}^\intercal $ \\[0.1cm]
        $\l = \:
            \begin{bmatrix}
                 0 & \:0 & \:\: 0 &\:\: -1 &\: 0 &\: 0 & \: 0 & \: 0 &\: 4
            \end{bmatrix} ^\intercal$ 
           
\end{tabular}

To obtain $A$, we first get all equations as in ~(\ref{eq:tableau_lpe}): 
$ 2x_1 + x_2 - x_{z(v_1)}  = 0$, $x_2 - x_1 - x_{z(v_2)}  = 0$, 
$2x_{f(v_2)} - x_{f(v_1)} -  y   = 0$, 
$x_{f(v_1)} - x_{z(v_1)}- x_{aux_1} = 0$, 
$x_{f(v_2)} - x_{z(v_2)}- x_{aux_2}  = 0$.

\noindent The first equation suggests we have coefficients: $2$ for $x_1$, $1$ for $x_2$ and $-1$ for $z(v_1)$. Because this equation does not feature any other variables in $\x$, we record zero coefficients for all other variables; which gives the first row in the matrix $A$. We continue in the same manner for the remaining equations:
 
\begin{tabular}{ c c } 

 $A =
  \begin{bmatrix}
    2 \:\:  & \: 1 \:  & -1  & 0  & 0  & 0  & \: 0 & 0  & \: 0 \:\\
    -1 \:\:  & \: 1 \:  &  0 & -1 & 0  & 0  & \: 0 \: & 0  & \: 0  \\
    0 \:\:  & \: 0 \:  &  0 &  0 & -1 & 2 & 0  & \: 0 & \: -1 \: \\
     0 \:\:  & \: 0 \:  &  -1 &  0 & 1 & 0 & -1  & \: 0 & \: 0 \: \\
      0 \:\:  & \: 0 \:  &  0 &  -1 & 0 & 1 & 0  & \: -1 & \: 0 \: \\
    \end{bmatrix} $

\end{tabular}

\noindent Finally, the set of \relu{} constraints is given by:
$\{x_{f(v_1)} = \relu (x_{z(v_1)}), x_{f(v_2)} = \relu (x_{z(v_2)}) \}$.
\end{example}

 \begin{definition}[Solution to DNN verification query]
    Given a vector $\sv \in \rn{n}$, variable assignment $\x = \sv$ 
    is a \emph{solution of the DNN verification query} 
    $\dnnvquery{A, \u, \l, C}$ if:
    \begin{equation}(A \cdot \x = \mathbf{0} \wedge \l \leq \x \leq \u \wedge \mathcal{C}(\x)) [\x / \sv] 
    \end{equation}
    We define the predicate 
    $\satpred{ \dnnvquery{A,\u, \l,C}, \sv}$  that returns \emph{true} if $\sv$ is a solution to $\dnnvquery{A, \u, \l, C}$.
    In case $C$ is empty, we will write $\satpred{ \dnnvquery{A,\u, \l}, \sv}$.

\end{definition}

\subsection{DNN Verification in Marabou}
\label{sec:dnnvmarabou}
In order to find a solution for a DNN verification query, or conclude there is none, DNN verifiers such as Marabou~\cite{WuIsZeTaDaKoReAmJuBaHuLaWuZhKoKaBa24} rely on algorithms such as Simplex~\cite{Da63} for solving the linear part of the query (i.e. $\dnnvquery{A,\u,\l})$, and enhance it with a splitting approach for handling the piecewise linear constraints (i.e. \relu{}) in $C$. Conceptually, the DNN verifier invokes the linear solver, which either returns a solution to the linear part of the query or concludes there is none. If no solution exists, the DNN verifier concludes the whole query is \unsat. If a solution is provided, then the verifier checks whether it satisfies the remaining \relu{} constraints. If so, the verifier concludes \sat{}. If not, case splitting is applied, and the process repeats for every case. For constraints of the form $f = \relu(z)$, a case-split divides the query into two subqueries: one enhances the linear part of the query with $f = 0 \wedge  u_z = 0$  and the other with  $aux = 0 \wedge l_z = 0 $. Adding $aux = 0$ is equivalent to adding $x_f = x_v$, thus by adding the auxiliary variable, case splitting is performed by updating the bounds ($\u_{aux} = \l_{aux} = 0$) without adding new equations. 
Besides splitting over a piecewise linear constraint, splitting can be performed over a single variable $x$ whose value can either be less or equal than some constant  $c\in\mathbb{R}$ or greater or equal than the constant; i.e., one subquery is enhanced with $u_x = c$ and the other with $l_x = c$.
This induces a tree structure, with nodes corresponding to the splits. If the verifier answers a subquery, then its node represents a leaf, as no further splits are performed. 
A tree with all leaf subqueries are \unsat{} corresponds to an \unsat{} query, and inversely a tree with a single \sat{} leaf corresponds to a \sat{} query. Note that in every such \unsat{} leaf, unsatisfiability is deduced using the linear part of the query. 

Also, note that such a DNN verification scheme heavily relies on a solver for linear equations, which often manipulates the matrix $A$, as well as many optimizations and heuristics for splitting strategy. A proof checker avoids implementing any of those operations.

\section{Farkas Lemma for DNN Verification}\label{sec:check}

Recall that the idea of a proof checker assumes that we can obtain a proof witness that the checker can certify independently.
Since a solution to a DNN verification query is a satisfiability problem, a satisfying assignment serves as a straightforward (and independently checkable) proof witness of \sat{}.
However, providing a proof witness for the \unsat{} case was an open question, based on the NP-hardness of the problem~\cite{SaLa21}. The novelty of  Isac et al.'s result~\cite{IsBaZhKa22} was in showing
 that the Farkas lemma~\cite{Va96} can be reformulated constructively in terms of the DNN verification problem, and that the vectors in its construction  
can be then used to witness the \unsat~{proof}.
We will call this specialized constructive form of the Farkas lemma
the \emph{DNN Farkas lemma} (cf. \S~\ref{sec:verificationProofs}). 

To use the DNN Farkas lemma as part of the verified checker, we  need to prove its correctness in Imandra. Indeed, there exist \emph{Mathematical Components}~\cite{affeldt2022mathcomp} proofs of the lemma in its original form~\cite{AllamigeonK19, vass_repo}. 
However, those formalisations rely strongly on comprehensive MathComp matrix libraries.  Other ITP attempts~\cite{botteschFarkasLemmaMotzkin2019,bessonModularSMTProofs2011,passmoreACL2ProofsNonlinear2023} have suggested ways to prove the Farkas lemma directly in terms of systems of linear polynomial equations (abbreviated \emph{systems of l.p.e.} for short), 
 obtaining the original Farkas lemma as a corollary (cf. Figure~\ref{fig:Farkas}). 
We will call this form of the Farkas lemma \emph{the polynomial Farkas lemma}.  

\begin{figure}[t]
\vspace{-0.6cm}
\centering

\begin{tikzpicture}[thick,
    set/.style = {circle,
        minimum size = 3cm}]
   
\node[rectangle,draw] (A) at (-3,4) {original Farkas lemma};

\node[rectangle,draw = red, double] (D) at (3,4) {polynomial Farkas lemma (\unsat)};

\node[anchor=south west, shift={(0.3,0.3)}, color=red] (D_lemma) at (3,4) {\footnotesize Theorem~\ref{thm:generalizedFarkas}};

\node[rectangle,draw, dashed] (B) at (-3,2.8) {DNN Farkas lemma};

\node[rectangle,draw = red, double] (E) at (3,2.8) {DNN polynomial Farkas lemma (\unsat)};
\node[anchor=south west, shift={(0.3,0.3)}, color=red] (E_lemma) at (3,2.8) {\footnotesize Theorem~\ref{ithm:dnn_poly_farkas}};

\node[rectangle] (F) at (0,1.2) {\unsat of DNN verification query};
\node[anchor=west, shift={(0.3,0)}, color=red] (F_lemma) at (1.5,2) {\footnotesize Theorem~\ref{ithm:reductiontogfl}};

\draw [<->, dotted] (A) edge (B);
\draw [<->] (D) edge (A);
\draw [->, double, red] (D) -- (E);
\draw [->, double, red] (E) -- (F);
\draw [->, dashed] (B) edge (F);
\end{tikzpicture}

\caption{Soundness and known relations between different versions of the Farkas lemma. Boxes represent lemmas that are proven sound.  An arrow from $A$ to $B$ states that the soundness of $A$ implies the soundness of $B$. Dotted lines show conjectures; dashed lines show results that were proven manually;  solid lines show those formalized in an ITP. Single black lines indicates previously known results, double red lines refer to results presented here. We indicate the lemmas  that correspond to each result.}
\vspace{-0.6cm}

\label{fig:Farkas}
\end{figure}
We follow the polynomial approach and prove the polynomial Farkas lemma in Imandra (\unsat{} case only, cf. Theorem~\ref{thm:generalizedFarkas}), obtaining its specialization to the case that refers to the parameters of the DNN verification problem, referred to as the \emph{DNN Polynomial Farkas lemma} (cf. Theorem~\ref{ithm:dnn_poly_farkas}).
In this shape, the lemma relates the existence of witnesses of a certain form 
to the unsatisfiability of a system of l.p.e. constructed from a DNN verification query.
To ensure that the checker is sound,
it remains to link this unsatisfiability to the unsatisfiability of the DNN verification query (Theorem~\ref{ithm:reductiontogfl}). 
 These results are shown in red double lines in Figure~\ref{fig:Farkas}. 

\subsection{The DNN Farkas Lemma}
\label{sec:verificationProofs}

To produce proofs of \unsat{} for linear queries of the form $\dnnvquery{A,\u, \l}$, Isac et al.~\cite{IsBaZhKa22} prove a constructive variant of the Farkas lemma 
that defines proof witnesses of \unsat{}:

\begin{theorem}
\label{thm:mainthm}[DNN Farkas lemma~\cite{IsBaZhKa22}]
    	Let $ A \in\R^{m \times n}$,  $\x$ a vector of $n$ variables and $ \l, \u \in \mathbb{R}^n $, such that $ A\cdot \x = \mathbf{0} $ and $ \l \leq \x \leq \u$.
	Then exactly one of these two options holds:
	\begin{enumerate}
		\item \sat{}: There exists a solution
		$ \sv \in \mathbb{R}^n $ such that    $(A \cdot \x = \mathbf{0} \wedge \l \leq \x \leq \u ) [\x / \sv]$ is true. 
		\item 
  \unsat{}: There exists a \emph{contradiction vector} $ w \in \mathbb{R}^m $ such
		that for all $\x$, if $ \l \leq \x \leq \u $ we have $ w^\intercal
                \cdot A \cdot \x < 0$. As $ w \cdot \mathbf{0} = 0 $,
                $w$ is a proof of the constraints' unsatisfiability. 
                       \end{enumerate}
 Moreover, in the \unsat{} case, the contradiction vector can be constructed during the execution of the Simplex algorithm.
\end{theorem}
\vspace{-0.2cm}

Note that, given a vector $s$ or $w$, a proof checker can immediately conclude the satisfiability of $\dnnvquery{A, \u, \l}$, while constructing $s$ or $w$ often requires the DNN verifier to apply more sophisticated procedures~\cite{IsBaZhKa22}. We will prove a polynomial version of this result in Imandra.

\subsection{Polynomial Farkas Lemma}
\label{sec:generalizedFarkas}

Given \emph{coefficients} $\overline{\alpha} = \alpha_0, \ldots , \alpha_{n} \in \rn{n + 1}$ and a vector of variables $\x = (x_1,...,x_n)$, a \emph{real linear polynomial}  $p^{\overline{\alpha}}(\vars)$  of size $n$ is 
defined as follows:
  $  p^{\overline{\alpha}}(\vars) := \alpha_0 + \sum\limits_{i=1}^n \alpha_i x_i. $
\noindent Given $\vars$, one can form different polynomials $p^{\overline{\alpha}_1}(\x), \ldots , p^{\overline{\alpha}_m}(\x)$. To simplify the notation, we write $p_1(\x), \ldots ,p_m(\x)$ to denote $m$ arbitrary polynomials when the concrete choice of $\overline{\alpha}_i$s is unimportant. We will use $q_1(\x), \ldots ,q_m(\x)$ to denote a potentially different list of polynomials.   
A \emph{linear polynomial equation} is an equation of the form $p(\x) \bowtie 0$, where $\bowtie$ is either $=$ or $\geq$.
The following formula defines a  \emph{system  of linear polynomial  equations}   $\mathcal{S}(\x) := (\bigwedge\limits_{i = 1}^K (p_i(\x) = 0)) \land (\bigwedge\limits_{j = 1}^{M} (q_j(\x) \geq 0))$
 where $K, M$ are arbitrary natural numbers. 
 We say that $\mathcal{S}(\x)$ is \emph{satisfied by a vector} $\sv \in \R^n$
if $\mathcal{S}(\x )[\x / \sv]$ is true. 
In Imandra, we use the predicate $\satpredsys{\mathcal{S}(\x),\sv}$, which returns \emph{true} when the vector $\sv$ is a solution to $\mathcal{S}(\x)$. 
We say the system of l.p.e. $\mathcal{S}(\x)$ \emph{has a solution}, denoted \sat, if there exists an $\sv$ such that $\mathcal{S}(\x )[\x / \sv]$ is true. Otherwise, $\mathcal{S}(\x )$ is \unsat.

\begin{figure}[t!]
\vspace{-0.2cm}
    \centering
    \begin{minipage}{1.07\linewidth}
    \begin{multicols}{2}
\begin{lstlisting}[style=imandra,morekeywords={@@by, @@fc, use, theorem, false, auto},emph={system, certificate, var_vect}]
theorem farkas_unsat (s: system) 
    (x: var_vect) (c: certificate) =
  well_formed s x && check_cert s c  
  ==>
  eval_system s x = false
[@@by [%use cert_is_neg s c x]             
   @> [%use solution_is_not_neg s c x]
   @> auto] [@@fc] \end{lstlisting}
        \columnbreak

        \begin{framed}
        \small{\textbf{Theorem 6.} (Polynomial Farkas Lemma (\unsat case)).
        \label{ithm:generalizedfarkas}
        \textbf{If}
           $\exists \mathbb{I}, \mathbb{C}. \mathbb{I}(\x) + \mathbb{C} (\x) < 0$ \textbf{then} 
                    $\neg(\exists \sv. \satpredsys{S(\x)), \sv}$ }
    \end{framed}    
    \end{multicols}
    
    \end{minipage}

    \caption{Left: Imandra code for the Polynomial Farkas lemma. The keywords \imlcode{use} and \imlcode{auto} are theorem proving tactics of Imandra; the \imlcode{fc} annotation instructs Imandra to automatically make use of the lemma as a \emph{forward chaining} rule.  Note the explicit use of lemmas \imlcode{cert_is_neg} and \imlcode{solution_is_not_neg} in the tactic script. Right: pseudocode formulation of the same theorem. In pseudocode, we assume all elements are well-formed, i.e. the  system of l.p.e. is valid w.r.t. $\x$. In Imandra, we assert this with the predicate \imlcode{well_formed}, and the function \imlcode{check_cert} checks whether the provided coefficients build witnesses $\mathbb{I}, \mathbb{C}$ s.t. $\mathbb{I} +  \mathbb{C}<0$. Then \imlcode{eval_system} checks if the assignment $\x$ is a solution to the  system of l.p.e.}
    \label{fig:farkas-imandra}
\end{figure}

\noindent\textbf{Operations on Polynomials.}
A polynomial $p^{\overline{\alpha}}(\x)$ with coefficients $\overline{\alpha} = \alpha_0, \ldots, \alpha_n$ \emph{scaled by a real constant} $c$ is defined as the polynomial over $\x$ with coefficients $c \alpha_0, \ldots , c \alpha_n$, and noted $c \cdot p(\x)$. The \emph{addition of two polynomials} $p^{\overline{\alpha}}(\x)$ and $p^{\overline{\beta}}(\x)$ with coefficients $\overline{\alpha} = \alpha_0, \ldots,\alpha_n$ and $\overline{\beta} = \beta_0,\ldots, \beta_n$ respectively, denoted as $p^{\overline{\alpha}} + p^{\overline{\beta}}(\x)$, is a polynomial with coefficients $\alpha_0 + \beta_0, \ldots, \alpha_n + \beta_n$.
The \emph{linear combination of polynomials $p_1(\x), \ldots ,p_N(\x)$ 
with coefficient vector $\c = c_1, \ldots, c_N \in \mathbb{R}^N$} is defined as the sum of polynomials $p_1(\x), \ldots ,p_N(\x)$ scaled by the coefficients in $\c$: $\sum\limits_{i=1}^N c_i p_i(\x)$.

We define the conversion function from the DNN verification query form to linear polynomial equations and inequalities.

\begin{definition}{Tableau-to-polynomial conversion.}\label{def:matrix_to_poly}
Given a DNN verification query $\dnnvquery{A, \u, \l}$ and variable vector $\x$, the conversion function $\matrixtopoly{A, \u, \l}$ is defined as follows:
\begin{equation}
\matrixtopoly{A, \u, \l}(\x) := \bigwedge_{i=0}^m p^{\alpha(i)}(\x) = 0 \land \bigwedge_{j=0}^n q^{\beta(\u, j)}(\x) \geq 0  
    \land \bigwedge_{k=0}^n q^{\gamma(\l, k)}(\x) \geq 0, \textrm{where}
\end{equation}

\begin{itemize}
    \item $p^{\alpha(i)}$ is a polynomial whose coefficients are the values in the $i^{th}$ row of the matrix $A$ (and constant element is $0$): $p^{\alpha(i)}(\x) := p^{0,a_{i,1},\ldots, a_{i,n}}(\x)$. It encodes the linear polynomial equations as defined in equation (\ref{eq:tableau_lpe}).
    \item For all $j \in [0, n]$, $q^{\beta(\u, j)}(\x) := \u_j - x_j$
    \item For all $k \in [0, n]$, $q^{\gamma(\l, k)}(\x) := x_k - \l_k$
\end{itemize}

Note that the coefficients for polynomials $q^{\beta(\u, j)}$ and $q^{\gamma(\l, k)}(\x)$ are sparse, with $0$ values at all indices except for the constant and at the index of $x_i$.
 \end{definition}

\begin{theorem}[Polynomial Farkas Lemma]
\label{thm:generalizedFarkas}
 
A  system of l.p.e. $\mathcal{S}(\vars) := \bigwedge\limits_{i = 1}^K(p_i(\vars) = 0) \land \bigwedge\limits_{j = 1}^M (q_j(\vars) \geq 0)$
has a solution iff there does not exist:
\textbf{1.} a linear combination $\mathbb{I}(\x)$ of $p_1(\x), \ldots, p_K(\x)$ and 
    \textbf{2.}  a linear combination $\mathbb{C}(\x)$ of $q_1(\x), \ldots, q_M(\x)$ with non-negative coefficients, 
 such that $\mathbb{I}(\x) + \mathbb{C} (\x) < 0$. We call $\mathbb{I}, \mathbb{C}$ \emph{witnesses} of \unsat{} for  $\mathcal{S}(\x)$.
\end{theorem}
\begin{proof}
We prove that witnesses $\mathbb{I}(\x),\: \mathbb{C}(\x)$ and a solution cannot exist simultaneously.
\begin{enumerate}
    \item Assume there exists
$c = (c_1, \ldots, c_{K})\in\rn{K}$ and $d = (d_1, \ldots, d_{M})\in \rn{M}$, where $d_{j} \geq 0$ for all $d_j \in d$, such that $\mathbb{I}(\x)+\mathbb{C}(\x)= \sum_{i = 1}^K c_{i} \cdot p_i(\x) + \sum_{j = 1}^M d_{j} \cdot q_j(\x) < 0$. 
\item Furthermore, assume that the  system of l.p.e. $\mathcal{S}(\x)$ has a solution $\sv$. Then,  for $\mathcal{S}(\x)[\x / \sv]$ and for all $i \in [1, K], p_i(\sv) = 0$, so we have $\sum\limits_{i = 1}^K c_{i} \cdot p_i(\sv) = 0$.
Similarly,  it implies that for all $j \in [1, M]$,  $q_j(\sv) \geq 0$, so we have $\sum\limits_{j = 1}^M d_{j} \cdot q_j (\sv) \geq 0$.
This means that $\sum\limits_{i = 1}^K c_{i} \cdot p_i(\sv) + \sum\limits_{j = 1}^M d_{j} \cdot q_j (\sv) \geq 0$.
\end{enumerate}

Since we postulated that this sum is negative, this leads to a contradiction.

\end{proof}

\begin{figure}
    \begin{minipage}{1.035\linewidth}
    \begin{multicols}{2}
\begin{lstlisting}[language=caml,xrightmargin=0pt,xleftmargin=1pt,basicstyle=\scriptsize\ttfamily,morekeywords={Goal,SubgoalEnter,false,true,H1,H0},emph={list}]
Goal:
((not
 ((eval_system es x) && (well_formed es x) 
    && es <> [])
  ==>
  (eval_poly (mk_cert_poly cs es) x) >=. 0.0)
 || (not
 ((well_formed es x) && (check_cert es cs))
  ==>
  (eval_poly (mk_cert_poly cs es) x) <. 0.0)
 || ((well_formed es x) && (check_cert es cs)  && es <> [])
   ==>
  (eval_system es x) = false)
.
Enter waterfall with goal:
fun (cs : real list) (es : expr list) 
    (x : real list) 
-> ((not
 ((eval_system es x) && (well_formed es x) 
     &&  es <> [])
  ==>
  (eval_poly (mk_cert_poly cs es) x) >=. 0.0)
 || (not
 ((well_formed es x) && (check_cert es cs))
  ==>
  (eval_poly (mk_cert_poly cs es) x) <. 0.0)
   || ((well_formed es x) && (check_cert es cs) && es <> [])
   ==>
   (eval_system es x) = false)
        1 nontautological subgoal.
Subgoal 1:
H0.((eval_system es x) && (well_formed es x) && es <> [])
     ==>
     (eval_poly (mk_cert_poly cs es) x) >=. 0.0
H1.((well_formed es x) && (check_cert es cs))
     ==>
     (eval_poly (mk_cert_poly cs es) x) <. 0.0
|-------------------------------------------
 ((well_formed es x) && (check_cert es cs) && es <> [])
 ==>
 (eval_system es x) = false
  
	But simplification reduces this to true, using the forward-chaining rules eval_const_neg and scale_empty_invariant and times_neg and times_neg_2 and times_pd and times_psd
\end{lstlisting}

\end{multicols}
\end{minipage}

    \caption{Proof generated by Imandra for Theorem~\ref{thm:generalizedFarkas} (Figure~\ref{fig:farkas-imandra}). It is resolved by simplification and using auxiliary lemmas. Although the application of lemmas \imlcode{cert_is_neg} and \imlcode{solution_is_not_neg} (see Table~\ref{tab:lemmas}) does not appear  explicitly, the auxiliary lemmas shown in the proof trace stem from it.} 
    \label{fig:enter-label}
    \vspace{-0.4cm}
\end{figure}

The Imandra code covering one direction of Theorem~\ref{thm:generalizedFarkas} is given in Figure~\ref{fig:farkas-imandra}.
It is the only direction needed to prove soundness of our implementation, and corresponds to the top double box in Figure~\ref{fig:Farkas}.
In Imandra, we define a valid certificate as a linear combination of the tableau rows which gives a polynomial with all coefficients equal to $0$ and a negative constant. 
The proof then proceeds by applying two lemmas: \imlcode{cert_is_neg}
 (``a valid certificate always evaluates 
to a negative value'') and \imlcode{solution_is_not_neg}  (``evaluating a certificate with a solution to the system evaluates to a non-negative value''). Table~\ref{tab:lemmas} shows formal statements of these lemmas. The tactic \imlcode{auto} resolves the contradiction in the theorem's conclusion.
We require that the system and variable vectors have matching dimensions, and use  the predicate \imlcode{well_formed} to assert that. 
Figure~\ref{fig:enter-label} gives the interested reader a glimpse of the completed proof produced by Imandra in response to user prompt given in Figure~\ref{fig:farkas-imandra}.

Generally, the user interacts with Imandra by supplying a set of automation tactics and possibly providing missing auxiliary lemmas.
The keyword \imlcode{waterfall} refers to a famous Boyer-Moore inductive proof automation approach~\cite{boyerComputationalLogic1979} that is deployed in ACL2 and Imandra~\cite{passmoreImandraAutomatedReasoning2020}. It dynamically composes tactics that generate induction schemes, perform simplification, rewriting, generalization and forward-chaining reasoning during proof search. Together with proofs of auxiliary lemmas the proof of Theorem~\ref{thm:generalizedFarkas} is nearly 200 lines long.

\begin{table}[t]
\centering
\vspace{-0.4cm}
\scalebox{0.9}
{
\begin{tabular}
{
> {\raggedright\arraybackslash}p{0.5cm}
>{\raggedright\arraybackslash}p{3.5cm}>{\raggedright\arraybackslash}p{10.5cm}}
\toprule
\textbf{Sec.} & \textbf{Lemma Name} & \textbf{Mathematical Statement} \\
\midrule
\label{lem:cert_is_neg}
\S~\ref{sec:check} & \imlcode{cert_is_neg} & If $\mathbb{I}(\x) + \mathbb{C}(\x)$ is a polynomial with $0$ coefficients and a negative constant, then $\forall \sv. \mathbb{I}(\x) + \mathbb{C}(\x)[\x/\sv] < 0$ \\
& \imlcode{solution_is_not_neg}\label{lma:sol_not_neg} & $\forall \sv$ s.t. $\sv$ is a solution to $S$. $\mathbb{I}(\x) + \mathbb{C}(\x)[\x/\sv] \geq 0$ \\ 
 & \imlcode{tableau_reduction}\label{lem:tableau_reduction} & If $\exists \sv$ s.t. $A.\sv = \Vec{0}$, then $\sv$ satisfies the equalities in $\matrixtopoly{A, \u, \l}(\x)$ \\ 
\label{lem:bound_reduction}
& \imlcode{bound_reduction} & If $\exists\sv$ bounded by $\l$ and $\u$, then $\sv$ satisfies the inequalities in $\matrixtopoly{A, \u, \l}(\x)$\\ 
\midrule
\S~\ref{sec:checktree} & \imlcode{relu_split}\label{lem:relu_split} &     $\forall x_f, x_b, x_{aux} \in \R$ such that $x_f=\relu(x_b)$ and $x_{aux} = x_f - x_b$, either $(x_b \geq 0 \wedge x_{aux} = 0)$ or $(x_b < 0 \wedge x_f = 0)$ hold\\ 
& \imlcode{get_set_nth}\label{lem:get_set_nth} & $\forall \l \in \R^n, k \in \R, i \in [0,n-1]. \l[\l_i/k]_i = k$ \\
\label{lem:set_nth_unchanged}
& \imlcode{ith_set_nth} & $\forall \l \in \R^n, k \in \R, i,j \in [0,n-1]$ s.t. $i \neq j. \l[\l_i/k]_j = \l_j$ \\
%& (part of Imandra's theory of reals) & $\forall x, y \in \R, x < y \implies x \leq y$ \\
\midrule
\S~\ref{sec:checkerSoundness}\label{lem:well_formed_preservation}  & \imlcode{well_formed_preservation} & If $\dnnvquery{A,\u,\l}$ well-formed w.r.t. variable vector $\x$, then $\dnnvquery{A, \l^L, \u^L}$ and $\dnnvquery{A, \l^R, \u^R}$ are also well-formed w.r.t. $\x$. \\ 
\bottomrule
\end{tabular}
}
\caption{Summary of main lemmas used in proofs of Farkas lemma and Soundness. The notation in this table follows the one given in Theorem~\ref{thm:mainthm} and Definition~\ref{def:matrix_to_poly}: $A$ is a generic tableau in $\R^{n \times m}$; $\x, \l, \u$ are vectors in $\R^n$; $\matrixtopoly{A,\u,\l}$ is the corresponding system of l.p.e; $\l^L, \u^L, \l^R, \u^R$ are updated bounds computed by \imlcode{update_bounds}.}
\label{tab:lemmas}
 \vspace{-0.2cm}
\end{table}

\subsection{Farkas Lemma for DNN Proof Checking}
\label{sec:cores}

Marabou produces proofs of \unsat{} in the form given in Theorem~\ref{thm:mainthm}, but the Imandra implementation uses the formulation of the Farkas lemma given in Theorem~\ref{thm:generalizedFarkas}.
We need to prove that Theorem~\ref{thm:generalizedFarkas} specializes to systems of l.p.e that encode DNN verification queries:

\begin{theorem}[DNN Polynomial Farkas Lemma]
\label{ithm:dnn_poly_farkas}
  $\forall A, \l, \u, C$ \textbf{If} $\exists \mathbb{I}, \mathbb{C}. \mathbb{I}(\x) + \mathbb{C} (\x) < 0$ \textbf{then} 
   $\neg(\exists \sv. \satpredsys{ S^{A, \l, \u}(\x),\sv )}$.
\end{theorem}

\noindent To deduce the \unsat{} of a DNN verification query from the \unsat{} of its equivalent system of l.p.e., we prove the following Lemma (represented by the lower red double arrow in Figure~\ref{fig:Farkas}). 

\begin{theorem}[Sound Application of DNN Polynomial Farkas Lemma]
    \label{ithm:reductiontogfl}
 \textbf{If} \\
  $  \neg(\exists \sv. \satpredsys{\matrixtopoly{A, \u, \l}(\x)}, \sv)$
  \textbf{then}
    $\neg(\exists \sv. \satpred{\dnnvquery{A, \u, \l}, \sv} $.
\end{theorem}

It is straightforward to see that, by Definition~\ref{def:matrix_to_poly}, polynomials in $\matrixtopoly{A,\u,\l}$ have a one-to-one correspondence to constraints in $\dnnvquery{A,\u,\l}$. However proving this in Imandra was not trivial. We consider the equations from the tableau rows and the inequalities from the bounds separately. For the equations,  \imlcode{tableau_reduction} is proven by Imandra with minimal guidance by induction on the tableau length. On the other hand, \imlcode{bound_reduction} necessitated linking index-based computations (the bound polynomial for the $i^{th}$ bound is constructed as a polynomial with a single non-zero coefficient at index $i$) and recursion-based computation (e.g. for checking whether a vector is bounded by $\l$ and $\u$).

Assembling Theorems~\ref{thm:generalizedFarkas} and~\ref{ithm:reductiontogfl}, we get that a witness of \unsat{} for a system of l.p.e. constructed from a DNN verification query is sufficient to guarantee that the query is \unsat{}.

\section{A DNN Certificate Checker in Imandra}
\label{sec:imandra_checker}\label{sec:checktree}

Recall that Marabou search induces tree-like structures. 
The proof production in~\cite{IsBaZhKa22} constructs a \emph{proof tree} based on the search trace of \unsat{} queries, serving as a witness. 
We now introduce an alternative algorithm that checks 
these proof trees and certifies that the corresponding DNN verification query is \unsat. If an error is detected, it can identify the specific parts of the proof that contributed to the failure. These correspond to search states where Marabou failed to produce a correct proof.

\begin{example}[Marabou Proof Tree Construction]\label{ex:pt}
Below is a graphical representation of a proof tree proving that the query presented in Example~\ref{ex:query} is \unsat{} and witnessing that during execution, Marabou performed a single split over the \relu{} constraint derived from $v_2$.  
The proof contains the tableau $A$, the bound vectors $\u,\l$, and the \relu{} constraints $C$:

	\centering
	\scalebox{0.83} {
		\def\xSep{2cm}
		\def\ySep{1.2cm}
		\begin{tikzpicture}[ >=stealth,shorten >=1pt,shorten <=1pt]
			
            \node[proofNode] at (0,0) (root) [] {$\dnnvquery{A,\u, \l,C}$};
   
			% \node[proofNode] (splits) at (0, -1*\ySep)  [] {$), \:\:\:\: ($};
            
            \node[proofLeaf] (active)[below=0.5cm of root, xshift=1.4*\xSep] {$(aux_2=0)\wedge (\l_{z(v_2)}=0)$};
   
			\node[proofLeaf](inactive)[below=0.5cm of root, xshift=-1.4*\xSep] {($f(v_2)=0)\wedge (\u_{z(v_2)}= 0)$};

			\node (leftProof)[below=0.1cm of active] {$\begin{bmatrix} 0&0&1&-2&0\end{bmatrix}^\intercal$};
   
			\node(rightProof)[below=0.1cm of inactive] {$\begin{bmatrix} 0&0&1&0&0\end{bmatrix}^\intercal$};

             \draw[proofEdge] (root) --node[label={[xshift=-0.4cm,yshift=-0.1cm]}] {} (active);

			\draw[proofEdge] (root) --node[label={[xshift=0.4cm,yshift=-0.1cm]}] {} (inactive);
		\end{tikzpicture}
	}

Each leaf contains a contradiction vector for their respective subquery (cf. Theorem~\ref{thm:mainthm}). 
\end{example}

\begin{figure}[t]
\vspace{-0.5cm}
    \centering
\begin{lstlisting}[style=imandra,morekeywords={Leaf,Node,Split,Relu,SingleVar},moreemph={type,of,int,list}]
type proofTree = 
  | Leaf of real list
  | Node of Split * proofTree * proofTree

type Split = Relu of int * int * int
  | SingleVar of int * real
\end{lstlisting}

        \vspace{-0.4cm}
	\caption{Marabou proof tree as an inductive data type in Imandra.}
    \label{fig:ipt}
    \vspace{-0.4cm}

\end{figure}

A proof tree \imlcode{t}, represented by an inductive type \imlcode{proofTree}, can be either a  leaf \imlcode{Leaf w} or a non-leaf node \imlcode{Node s tl tr}.
A leaf \imlcode{Leaf w}  corresponds to an \unsat{} result of the linear part of a subquery; 
leaves hence always contain a contradiction vector \imlcode{w : real list}.
A non-leaf node  \imlcode{Node s tl tr} corresponds to a split \imlcode{s: Split} and corresponding sub-proof-trees \imlcode{tl} and \imlcode{tr}. 
We consider two kinds of splits: 
\imlcode{s: Split} can be either a case split over a \relu constraint,
\imlcode{Relu b f aux}, where \imlcode{b}, \imlcode{f}, and \imlcode{aux} are the participating variables in the corresponding \relu constraint (i.e. the DNN verification query includes $\x_f = \relu(\x_b) \wedge  \x_f - \x_b - \x_{aux} = 0$), 
or a single-variable split \imlcode{SingleVar xi k}, performed on a variable $\x_i$ with a split on value $k$.

This type guarantees some structural properties but not the full certification of the proof trees. For example, for \imlcode{Leaf w}, \imlcode{w} has to be checked to be a contradiction vector for the corresponding subquery to construct a witness of \unsat (Theorem~\ref{thm:generalizedFarkas}).
During the parsing of the Marabou proof, we hence check that each \relu split corresponds to one of the constraints in the original DNN verification query 
via the recursive \checknode{} algorithm (Algorithm~\ref{alg:proofchecker}).

Marabou uses splits to compute tighter upper and lower bound vectors, 
later used to prove \unsat in the proof tree's leaves. 
The function \imlcode{update_bounds u l s} performs this bound tightening and yields updated lower and upper bounds for both subtrees:

\begin{lstlisting}[style=imandra, numbers=left,numberstyle=\tiny,morekeywords={SingleSplit,ReluSplit,Split},moreemph={let,match,with,in,list}]
let update_bounds (lbs: real list) (ubs: real list) (split: Split): 
        ((real list * real list) * (real list * real list)) =
    match split with
    | SingleSplit (i,k) -> ((lbs,set_nth ubs i k), (set_nth lbs i k,ubs))
    | ReluSplit (b, f, aux) ->
        let lbs_l = set_nth lbs f 0. in
        let ubs_l = set_nth (set_nth ubs b 0.) f 0. in
        let lbs_r = set_nth (set_nth lbs b 0.) aux 0. in
        let ubs_r = set_nth ubs aux 0. in
        (lbs_l, ubs_l), (lbs_r, ubs_r)
\end{lstlisting}

   In the case of a split on a single variable $i$ with value $k$, the updated bounds are $\l,\u[\u_i/k]$ and $\l[\l_i/k],\u$. 
    In the case of a split on a \relu constraint of the form $f = \relu (b) \wedge aux = f - b$, the updated bounds are $\l[\l_b/0,\l_{aux}/0]$ and $\u[\u_{aux}/0]$ for one child and $\l[\l_f=0]$ and $\u[\u_b/0, \u_f/0]$ for the other. This corresponds to the case analysis of lemma \imlcode{relu_split} (see Table~\ref{tab:lemmas}). We call the first case \emph{inactive phase} and the second case \emph{active phase} of the \relu constraint.
 Note that the order of the phases is fixed during the tree construction; the fixed order is followed in the bound updating implementation (see lines 4, 12).

The function \checknode$(A, \u, \l, C, \texttt{t} )$, 
traverses a proof tree \imlcode{t} and given a DNN verification  query $\dnnvquery{A, \u, \l, C}$ ensures that for all the leaves it is able to build witnesses of \unsat{} in the sense of Theorem~\ref{ithm:dnn_poly_farkas}.
It is given in pseudocode (Algorithm~\ref{alg:proofchecker}) and Imandra code (Figure~\ref{lst:checknode}).
For a leaf node containing a list of reals $w$, the procedure \imlcode{mk_certificate} computes witness candidates $\mathbb{I(\x), C(\x)}$ as in Theorem~\ref{ithm:dnn_poly_farkas}, then the algorithm checks whether $\mathbb{I}(\x) + \mathbb{C}(\x) < 0$ (\imlcode{mk_certificate} in Figure~\ref{lst:checknode}, l. 6). If this check passes, by  Theorem~\ref{thm:generalizedFarkas}, $\mathbb{I(\x), C(\x)}$ are witnesses of \unsat{} for the system of l.p.e. $\mathcal{S}^{A, \u, \l}(\x)$. By Theorem~\ref{ithm:reductiontogfl},  $\dnnvquery{A, \u, \l}$ is \unsat as well.
For a non-leaf node with data $\splt, \node^L, \node^R$:
        \begin{itemize}
        \item the procedure $\updatebounds$ computes updated versions of $\u, \l$ corresponding to the two phases of the split as described in the previous section (\updatebounds, lines~10-11).
        \item recursive calls to \checknode, using the new bounds that correspond to each sub-tree, check the sub-trees rooted in $\node^L$ and $\node^R$ (lines~13-14).
\end{itemize}
A tree $\node$ is \emph{certified}  if all checks pass. The function $\checknode(A, \u, \l, C, \node)$ implements Algorithm~\ref{alg:proofchecker} almost verbatim, with the notable difference in the use of the \imlcode{check_split} function (Figure~\ref{lst:checknode}, line~9). 
It re-iterates the check that \relu splits indeed correspond to a known \relu constraint of the verification query. 

\begin{algorithm}[t]
\small{
\caption{\small{Proof tree checking algorithm (\checknode)}}
\label{alg:proofchecker}
\begin{algorithmic}[1]
    \State Inputs: $A, \u, \l, C, \node$
    \State Outputs: $\{\emph{true, false}\}$
    \If{$\node$ is a leaf, with list of reals $w$}
        \State $\mathbb{I,C} := \mkcert{A, \u, \l, w}$
        \If{$\mathbb{I(\x) + C(\x)} < 0$}
        \State \Return \emph{true}
        \Else
        \State \Return \emph{false}
        \EndIf
    \EndIf
    \If{$\node$ is a non-leaf node, containing $\splt, \node^L, \node^R$}
        \State $\langle \lnode{L}, \unode{L}, \lnode{R},\unode{R}\rangle := \updatebounds(\u, \l, \splt)$
        \If{$\checknode(A, \lnode{L}, \unode{L}, C, \node^L) \newline 
        \hspace*{2em}\wedge \checknode(A, \lnode{R}, \unode{R}, C, \node^R)$}
        \State \Return \emph{true}
        \Else
        \State \Return \emph{false}
        \EndIf
    \EndIf
\end{algorithmic}}
\end{algorithm}

\begin{figure}
\begin{lstlisting}[style=imandra,xleftmargin=\parindent,numbers=left,numberstyle=\tiny,language=caml,morekeywords={Proof_tree,Leaf,Constraint,Node},emph={let,rec,with,in,list,match,bool}]
let rec check_node (tableau: expr list) (upper_bounds: real list) 
  (lower_bounds: real list) (constraints: Constraint.t list) 
  (proof_node: Proof_tree.t): bool =
    match proof_node with
    | Proof_tree.Leaf (contradiction, bound_lemmas) ->
        check_contradiction contradiction tableau upper_bounds 
        lower_bounds
    | Proof_tree.Node (split, bound_lemmas, left, right) ->
        let valid_split = check_split split constraints in
        let (lb_left, ub_left), (lb_right, ub_right) = update_bounds 
            lower_bounds upper_bounds split in
        let valid_children = 
            (check_node tableau ub_left lb_left constraints left) &&
            (check_node tableau ub_right lb_right constraints right) in
        valid_split && valid_children
\end{lstlisting}
\caption{Imandra implementation of \checknode}\label{lst:checknode}
\vspace{-0.4cm}
\end{figure}

The following example shows the execution of this algorithm.
\vspace{-0.3cm}
\begin{example}[Checking the Proof Tree of Example~\ref{ex:pt}. 
]
\label{ex:proof_checking}

To certify the proof tree, 
 the algorithm begins by checking the root node by certifying that the bound vectors of the two children correspond to the splits $\u_{f(v_2)} = 0 \wedge \l_{f(v_2)} = 0 \wedge \u_{z(v_2)} = 0  $ and $\u_{aux_2} = 0 \wedge \l_{aux_2} = 0 \wedge \l_{z(v_2)} =0$, i.e., to the two splits of the constraint $x_{f(v_2)}=\relu(x_{z(v_2)})$.
 Then, it recursively checks the leaves.
It starts by updating $\u_{z(v_2)}=0$, $\l_{f(v_2)} = 0$, and $\u_{f(v_2)} = 0$ and certifying the contradiction vector of the leaf: 
it constructs the system of l.p.e.~with equations representing the tableau $A$: $2x_1 + x_2 - x_{z(v_1)} = 0,\:\:\: x_2 - x_1 - x_{z(v_2)} = 0 ,\:\: \: 2x_{f(v_2)} - x_{f(v_1)} - y = 0, \:\:\: x_{f(v_1)} - x_{z(v_1)}- x_{aux_1} = 0,  \:\:\: x_{f(v_2)} - x_{z(v_2)}- x_{aux_2} = 0 $,
and the inequalities representing the bounds $\l,\u$. For simplicity, we consider only the inequalities $0 - x_{f(v_2)}  \geq 0, x_{f(v_1)} - 0 \geq 0$ and $y - 4 \geq 0$.
Then, it constructs the certificate $\mathbb{I} = \begin{bmatrix}
    0&0&1&0&0
\end{bmatrix}^\intercal\cdot A \cdot x =   2x_{f(v_2)} - x_{f(v_1)} - y = 0$, $\mathbb{C} = 2( - x_{f(v_1)} ) +  x_{f(v_1)} + y - 4$. Lastly, it checks that indeed $\mathbb{I} + \mathbb{C} = 2x_{f(v_2)} - x_{f(v_1)} - y - 2 x_{f(v_1)} +  x_{f(v_1)} + y - 4 = -4 < 0$.

The right leaf is checked similarly.
Since all nodes pass the checks, the algorithm returns \emph{true}, and the certification of this proof tree is complete.
\end{example}

\vspace{-0.2cm}
\section{Proof of Soundness}
\label{sec:checkerSoundness}
This section presents the  proof of soundness of the algorithm introduced in the previous section:
    given a DNN verification query $\dnnvquery{A, \u, \l, C}$ and the corresponding tree witness of \unsat $\node$, if \checknode$(A, \u, \l, C, \node)$ returns \emph{true}, then there should exist no satisfying assignment for $\dnnvquery{A, \u, \l, C}$.

The proof follows a custom induction scheme on the structure of \checknode. 
We will first prove the base case, when the proof tree is a leaf (Lemma~\ref{ilma:leaf_checking}), then the induction step, when it is a node with children (Lemma~\ref{lem:splits}). The proof for leaves is straightforward thanks to the Farkas lemma. The case for non-leaf nodes is trickier, as it involves proving that splits fully cover the DNN verification query solution space. This can be done by proving that both single variable splits and \relu splits are covering (Lemmas~\ref{lem:single_var_split} and \ref{lem:splits} respectively).

We first 
prove that \checknode{} is correct for proof trees leaves:

\begin{lemma}[Leaf checking]
\label{ilma:leaf_checking}
 \textbf{If}   $\node$ is a leaf  \textbf{and} $\checknode(A, \u, \l, C, \node)$ returns $true$
    \textbf{then}
   $ \neg(\exists \sv. \satpredsys{\matrixtopoly{A, \u, \l}(\x), \sv})$.
\end{lemma}
\begin{proof}
 By the definition of \checknode{} and Theorems~\ref{thm:generalizedFarkas} and \ref{ithm:reductiontogfl}.
\end{proof}

\noindent For the inductive case, we first prove that tightening bounds according to splits is covering. We state the definition of boundedness as a conjunction instead of a universally quantified index to avoid instantiating an index and to allow better automation in Imandra.
\begin{definition}[Bound vectors]\label{def:bound_vectors}
Let $\u, \l, \x \in \R^n$. We say that $\x$ is bounded by $\l$ and $\u$ (and we write $\l \leq \x \leq \u$) if $\bigwedge\limits_{i = 0}^{n-1} \l_i \leq \x_i \leq \u_i$. 
\end{definition}

\noindent For single variable splits,
the following lemma is proven (we omit the proof as it is simple): 
\begin{lemma}[Single variable splits are covering]\label{lem:single_var_split}
Let $\u, \l, \x \in \R^n$; $i \in [0, n-1]$; $k \in \R$; $\l^L := \l$, $\u^L := \u[\u_i/k]$, $\l^R := \l[\l_i/k]$, $\u^R := \u$.

 \textbf{If} $\l \leq \x \leq \u$ 
 \textbf{then} $\l^L \leq \x \leq \u^L \vee
 \l^R \leq \x \leq \u^R$.
\end{lemma}

 \noindent For \relu splits, recall that for each split with indices $b, f, aux$, the DNN verification query includes $\x_f = \relu(\x_b) \wedge  \x_f - \x_b - \x_{aux} = 0$.

\begin{lemma}[\relu splits are covering]\label{lem:splits}
Let $\u, \l, \x \in \R^n$; $f, b, aux \in [0, n-1]$ such that $f \neq b, f \neq aux, b \neq aux$; $\l^L := \l[\l_f/0]$, $\u^L := \u[\u_b/0, \u_f/0]$, $\l^R := \l[\l_{aux}/0]$, $\u^R := \u[\u_b/0, \u_{aux}/0]$.

 \textbf{If} $\l \leq \x \leq \u \wedge \x_f = \relu(\x_b)
 \wedge \x_f - \x_b - \x_{aux}  = 0$
 \textbf{then}
 $\l^L \leq \x \leq \u^L \vee \l^R \leq \x \leq \u^R$.
\end{lemma}

\begin{proof}
    Assuming that $\l \leq \x \leq \u \wedge \x_f = \relu(\x_b)
     \wedge \x_f - \x_b - \x_{aux} = 0$ holds, we consider the cases given by \imlcode{relu_split}, see Table~\ref{tab:lemmas} for statements of auxiliary lemmas mentioned below:
    \begin{enumerate}
        \item Case 1: $\x_b < 0 \wedge \x_f = 0$.
        We prove that for all indices $i \in [0, n-1]$, $\l[\l_f/0]_i \leq \x_i \leq \u[\u_b/0, \u_f/0]_i$, by considering cases according to the value of $i$.
        % \knote{missing lemma: add to Tab.1? same in the case below}
        \begin{enumerate}
            \item if $i = b$:
            \begin{multline*}    
            \l[\l_f/0]_i 
            \overset{f \neq b, \mathtt{ith\mathunderscore set\mathunderscore nth}}= \l_i \overset{ass., Def.\ref{def:bound_vectors}}\leq \x_i 
            \overset{i=b}= \x_b 
            \overset{ass., \mathtt{leq\mathunderscore lt}}\leq 0 
            \overset{i=b, \mathtt{get\mathunderscore set\mathunderscore nth}}= \u[\u_b/0, \u_f/0]_i
            \end{multline*}
            
            \item if $i = f$: 
            \begin{multline*}    
            \l[\l_f/0]_i \overset{i=f, \mathtt{get\mathunderscore set\mathunderscore nth}}= 0 \overset{ass.}\leq \x_i
            \overset{i=f}= \x_f 
            \overset{ass.}\leq 0
            \overset{i=f, \mathtt{get\mathunderscore set\mathunderscore nth}}= \u[\u_b/0, \u_f/0]_i
            \end{multline*}

            \item if $i \neq b \wedge i \neq f$: 
            \begin{multline*}
            \l[\l_f/0]_i 
            \overset{i\neq f,\mathtt{ith\mathunderscore set\mathunderscore nth}}= \l_i 
            \overset{ass., Def.\ref{def:bound_vectors}}\leq \x_i 
            \overset{ass., Def.\ref{def:bound_vectors}}\leq \u_i 
            \overset{i\neq f \wedge i\neq b, \mathtt{ith\mathunderscore set\mathunderscore nth}}= \u[\u_b/0, \u_f/0]_i
            \end{multline*}
        \end{enumerate}
\item Case 2: $\x_b \geq 0 \wedge \x_{aux} = 0$.
        We prove that for all indices $j \in [0, n -1]$, $\l[\l_b/0, \l_{aux}/0]_j \leq \x_j \leq \u[\u_{aux}/0]_i$, we do a similar case analysis on the value of $j$. The proof proceeds similarly to Case 1. 
\end{enumerate}
\vspace{-0.5cm}
\end{proof}
\vspace{-0.1cm}  
Lemma~\ref{ilma:split_soundness} states that if the query corresponding to a parent node in the proof tree has a satisfying assignment, then this assignment will satisfy one of the child node's queries. 

\begin{lemma}[Splits are covering]\label{ilma:split_soundness}
Let $\u, \l, \x \in \R^n$; $Split$ a split, $\l^L, \u^L$, $\l^R,\u^R$ are the bounds computed by $\updatebounds(\u,\l,Split)$.

 \textbf{If} there exists  $\sv$  that satisfies $\dnnvquery{A,\u,\l,C}$
 \textbf{then}
 $\sv$ satisfies $\dnnvquery{A,\u^L,\l^L,C}$ \textbf{or} 
 $\sv$ satisfies $\dnnvquery{A,\u^R,\l^R,C}$.
\end{lemma}
\begin{proof}
Since the bounds are the only parts of the children's queries that differ from their parent's query, proving Lemma~\ref{ilma:split_soundness} only requires to prove that the updated bounds are covering. This follows from Lemmas~\ref{lem:single_var_split} and \ref{lem:splits}.
  \vspace{-0.1cm}  
\end{proof}

Lemmas~\ref{ilma:leaf_checking} and \ref{ilma:split_soundness} now allow us to prove the overall soundness of \checknode.
\begin{theorem}
    [Algorithm~\ref{alg:proofchecker} is sound]
\label{ilma:checker_soundness}
\textbf{If}    $\checknode(A, \u, \l, C, \node) $ \textbf{returns true, then}\\
    $\neg(\exists \sv. \satpred{\dnnvquery{A, \l, \u, C}, \sv})
    $.
\end{theorem}

\begin{proof}
    We proceed by a functional induction scheme based on \checknode's definition.
    
    \textbf{Base case}: $\node$ is a leaf. By Lemma~\ref{ilma:leaf_checking}.
    
    \textbf{Induction step}: Let $Split$ be a split, $\node^L$ and $\node^R$ be two proof trees and $\u^L, \l^L, \u^R, \l^R$ be the bounds obtained from $\updatebounds(\l,\u,Split)$.
    Our induction hypothesis states that $\checknode$ is sound for $\checknode(A, \lnode{L}, \unode{L}, C, \node^L)$ and $\checknode(A, \lnode{R}, \unode{R}, C, \node^R)$. We now need to prove that $\checknode$ is sound for $\checknode(A, \l, \u, C, \node)$, where $\node$ is the proof tree with children $\node^L$ and $\node^R$ and split $Split$.
    
    Assuming that $\checknode(A,\l,\u,C, \node)$ holds, by definition  $\checknode(A, \lnode{L}, \unode{L}, C, \node^L)$ and $\checknode(A, \lnode{R}, \unode{R}, C, \node^R)$ also hold.
    By the induction hypothesis, this means that $\neg(\exists \sv. \satpred{\dnnvquery{A, \l^L, \u^L, C}, \sv})$ and $\neg(\exists \sv. \satpred{ \dnnvquery{A, \l^R, \u^R, C}, \sv})$

    By Lemma~\ref{ilma:split_soundness}, we conclude that $\neg(\exists \sv. \satpred{\dnnvquery{A, \l, \u, C}, \sv})$.
\end{proof}

\noindent Note that the proof assumes that all the query elements are well-formed w.r.t the variable vector, i.e. the tableau, bounds and variable vector dimensions match, and the proof-tree splits correspond to the query constraints. In the implementation, we need to prove that these properties are preserved throughout the inductive steps, with lemmas such as \imlcode{well_formed_preservation} (see Table~\ref{tab:lemmas}). 
Even though we need to give such indications, the custom inductive scheme is derived automatically by Imandra: the user tactic interaction with Imandra is shown in Figure~\ref{fig:final-proof}.

\begin{figure}[t!]
\vspace{-0.2cm}
    \centering
\begin{lstlisting}[basicstyle=\scriptsize\ttfamily,xleftmargin=0pt,xrightmargin=0pt,linewidth=1.05\linewidth,morekeywords={theorem,@@by,induction,expand,use,@@disable,Constraint,Proof_tree,},moreemph={list}]
theorem check_node_soundness (tableau: real list list) (upper_bounds: real list)
		(lower_bounds: real list) (constraints: Constraint.t list) (tree: Proof_tree.t) 
		(x: real list) =
    valid_proof tableau upper_bounds lower_bounds constraints tree
    && well_formed_vector tableau x
    ==> unsat tableau upper_bounds lower_bounds constraints x
[@@by [%expand "valid_proof"] 
   @> [%expand "well_formed_vector"] 
   @> induction ()
   @>>| [%use check_node_soundness_full tableau upper_bounds lower_bounds constraints tree x]
   @> [%use check_node_parent_imply_check_node_children (mk_eq_constraints tableau) 
            upper_bounds lower_bounds constraints tree]
   @> [%use well_formed_preservation tableau upper_bounds lower_bounds (split_of_node tree)]
   @> auto]
[@@disable List.length, well_formed_tableau_bounds, check_node, mk_eq_constraints, unsat, set_nth, update_bounds_from_split]
     \end{lstlisting}
     \vspace{-0.6cm}
      \caption{Proof of Theorem~\ref{ilma:leaf_checking} in Imandra: \imlcode{valid_proof} calls \checknode, in addition to some structural checks (e.g. on the tableau dimensions).}
    \label{fig:final-proof}
    \vspace{-0.5cm}
\end{figure}

\section{Evaluation}
\label{sec:eval}

In this section, we evaluate the new implementation of  across two orthogonal axes: first, the code complexity of its main modules; and second, the performance speed, compared to the Marabou C++ implementation. 

\begin{table}[ht]
\centering
\vspace{-0.4cm}
\scalebox{0.9}
{
\begin{tabular}{> {\raggedright\arraybackslash}p{3.0cm} >{\raggedright\arraybackslash}p{3.7cm} >
{\centering\arraybackslash}p{0.7cm} >
{\centering\arraybackslash}p{0.5cm}>
{\raggedright\arraybackslash}p{4.2cm}}
\toprule
\textbf{Result} &\textbf{Module name} &  \textbf{L.O.C} & \textbf{Aux. Lem.} &\textbf{Library Dependencies (accumulating) }  \\
\midrule
Poly Farkas lemma (\textbf{Theorem~\ref{thm:generalizedFarkas}}) & farkas.iml & 194 & 21 & Imandra Standard Libraries: Real, List, Polynomials  \\
\midrule
Sound application of & well\mathunderscore formed\mathunderscore reduction.iml& 41  & 13 & farkas.iml, certificate.iml,  \\
 DNN polynomial & bound\mathunderscore reduction.iml & 359 & &arithmetic.iml, util.iml,  \\
 Farkas lemma (\textbf{Theorem~\ref{ithm:reductiontogfl}})& 
 tableau\mathunderscore reduction.iml & 356 &  &   tightening.iml, constraint.iml, proof\_tree.iml, checker.iml, bound\_reduct\_g.iml, mk\_bound\_poly.iml \\
\midrule
Soundness of leaf checking (\textbf{Lemma~\ref{ilma:leaf_checking}}) & leaf\mathunderscore soundness.iml & 145 & 40 & sat.iml, split.iml 
bound\_reduction.iml, well\_formed\_reduction.iml, tableau\_reduction.iml
\\
\midrule
Single variable splits are covering (\textbf{Lemma~\ref{lem:single_var_split}) }&single\mathunderscore var\mathunderscore  split\mathunderscore soundness.iml & 81 & 10  &  \\
\midrule
 \relu splits are &relu\mathunderscore  split\mathunderscore soundness.iml & 113  & 18 & relu.iml  \\
covering &relu\mathunderscore  case\mathunderscore 1\mathunderscore bounded.iml & 337  & &   \\
(\textbf{Lemma~\ref{ilma:split_soundness}}) 
 &relu\mathunderscore  case\mathunderscore 2\mathunderscore bounded.iml & 338 & &  \\
\midrule
Soundness of node checking & node\mathunderscore soundness.iml & 78 & 19 & relu\_split\_soundness.iml, single\_var\_split\_soundness.iml  \\
\midrule
Soundness (\textbf{Theorem~\ref{ilma:checker_soundness}})&checker\mathunderscore soundness.iml & 71  & 146 &  leaf\_soundness.iml,
 node\_soundness.iml \\ 
\midrule
\midrule
\textbf{Total:} &  &\textbf{2113} &\textbf{267} & \\
\bottomrule
\end{tabular}
}
\caption{Summary of the entire formalisation. The Table reads as follows: a result $A$ is proven in module \texttt{A\_{mod}}, which is $N$ lines long, calls $M$ auxiliary lemmas, and depends on libraries as listed.}
\label{tab:lemmas2}
\vspace{-0.8cm}
\end{table}

\textbf{Code Complexity.} The overview of the entire formalisation is given in Table~\ref{tab:lemmas2}. In addition to counting L.O.C., we also count the number of auxiliary lemmas per proof, as they are the main mode of user interaction with Imandra proof search. 
We note that, due to the clever proof production offered by the Waterfall method in Imandra~\cite{boyerComputationalLogic1979,passmoreImandraAutomatedReasoning2020},
the size of the human-written code (as exemplified in Figure~\ref{fig:farkas-imandra}, \ref{fig:final-proof}) is much smaller than the actual length of the corresponding proof (as shown in Figure~\ref{fig:enter-label}).

\textbf{Performance Speed.}
As scalability of DNN verifiers is a major factor within the DNN verification community~\cite{BrBaJoWu24, BrMuBaJoLi23}, any implementation of algorithms should be considered with respect to its performance.
We used the proof producing version of Marabou to solve queries from two families of benchmarks: (1) \emph{collision avoidance} (coav), which verifies a DNN with 137 \relu neurons attempting to predict collisions of two vehicles that follow
curved paths at different speed~\cite{Eh17}; and (2) \emph{robotics navigation}~\cite{AmCoYeMaHaFaKa23} with properties of a neural robot controller with 32 \relu neurons. We chose these benchmarks because they provide a large dataset of \unsat{} DNN verification queries that are solvable in a short time. 
We have disabled some optimizations within Marabou --- proofs of bound tightenings that are derived based on the \relu{}  constraints, and thus are not proven directly by using the Farkas lemma~\cite{IsBaZhKa22}. 

\begin{figure}[b]
    \centering
    \vspace{-0.4cm}
    \begin{minipage}[b]{0.45\linewidth}
        \centering
        \includegraphics[width=\linewidth]{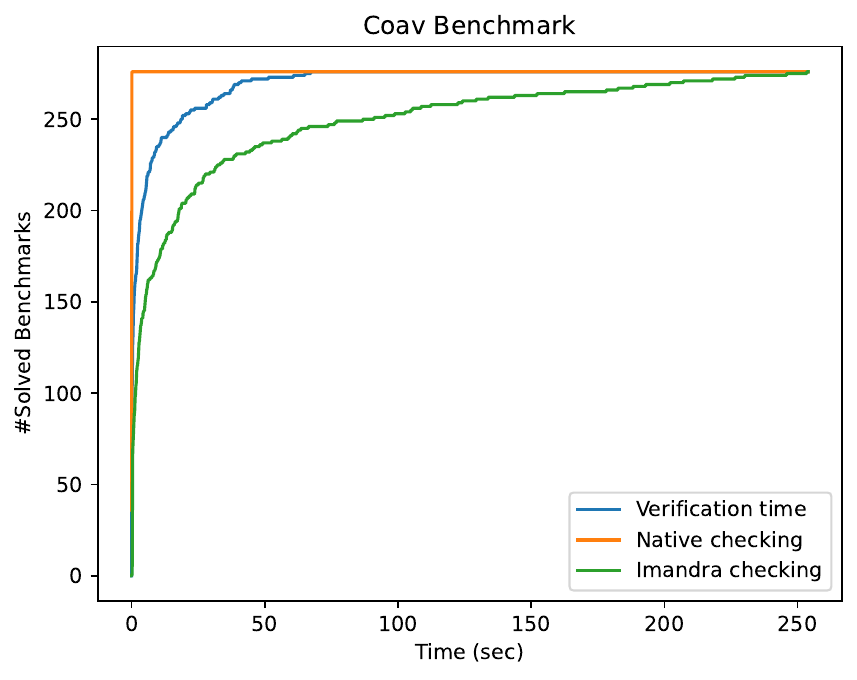}
    \end{minipage}
    \hfill
    \begin{minipage}[b]{0.45\linewidth}
        \centering
        \includegraphics[width=\linewidth]{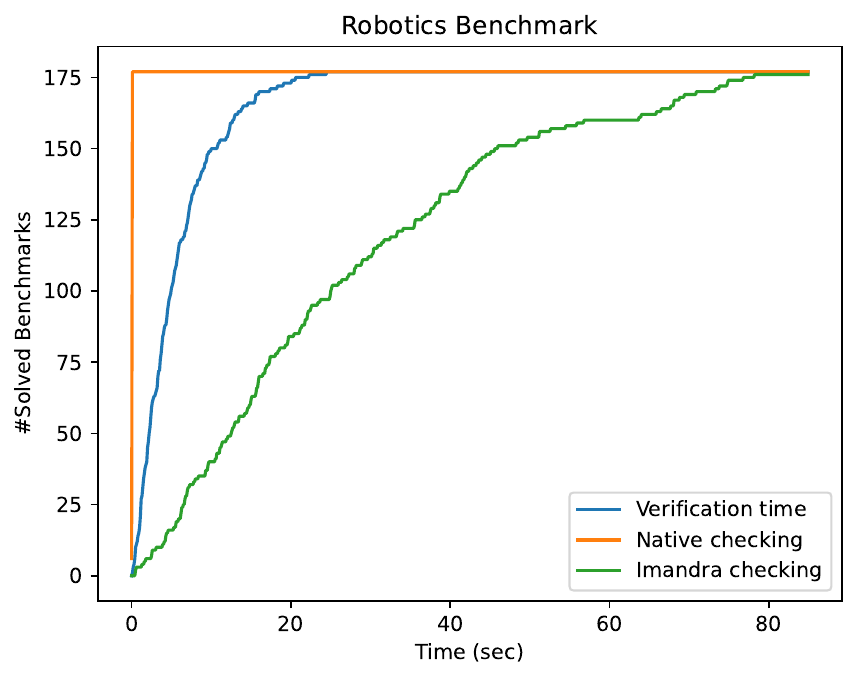}
    \end{minipage}
    \vspace{-0.4cm}
    \caption{Comparison of verification, native checking and Imandra checking time. Each point $(x,y)$ represents the number $x$ of verified/checked instances in $y$ seconds.}
    \label{fig:eval}
    \vspace{-0.4cm}
\end{figure}
Overall, Marabou has solved all queries and successfully generated \unsat{} proofs for 293 coav and 213 robotics queries. For some queries, proofs were not generated due to early \unsat{} deduction during preprocessing. We evaluated the results of 276 coav and 180 robotics queries, which have proof size smaller than 5MB.
All \unsat{} proofs were checked by the native Marabou checker and by Imandra, and both checkers have certified all queries. 
For each query, we measured the time it took Marabou to deduce \unsat{} (Verification time) and the time it took the native Marabou and Imandra to check the proofs (Native and Imandra checking time, respectively). Our evaluation is depicted in Figure~\ref{fig:eval}. 
For the coav and robotics benchmarks, on average, Marabou solved the queries in 5.40 and 5.62 seconds, respectively; while its native checker checked the proofs in average of 0.003, 0.03 seconds. Imandra required on average 24.64 and 26.72 seconds,  suggesting that, regardless of network size, Imandra requires time proportionate to the verification time of Marabou. This hypothesis needs to be further investigated. Furthermore, even though Imandra's checking time is considerably slower than Marabou's, its effect on the overall verification process is of factor $4.56\times-4.76\times$ slowdown.

Although performance differences are expected due to the use of more precise arithmetic, our results suggest a crucial trade-off between scalability and reliability. We believe that some performance difference is due to the use of verification-oriented data structures such as lists, and thus using optimized structures such as maps may improve speed significantly as shown in~\cite{desmartinNeuralNetworksImandra2022}, especially in the presence of sparse vectors in  DNN Verification queries.

\vspace{-0.2cm}
\section{Conclusions,  Future and Related  Work }
\label{sec:relatedwork}
The nascent field of DNN verifiers faces several challenges,
and certification of DNN verifiers themselves is widely recognized as one of them~\cite{PLNNV25,WuIsZeTaDaKoReAmJuBaHuLaWuZhKoKaBa24}. 
Recent work of~\cite{IsBaZhKa22} laid theoretical foundations for the certification of the DNN verifier Marabou, by connecting its certificate producing version with the well-known Farkas lemma. In this paper,  we take inspiration from the \emph{proof-carrying code}~\cite{Necula97,komendantskaya2025proofcarryingneurosymboliccode} and \emph{self-certifying code}~\cite{10.1145/3689624}  tradition and propose a framework  that implements, executes, and verifies a checker for certificates produced by Marabou, \emph{within the same programming language},
thus obtaining high assurance levels.
Moreover, thanks to Imandra's native use of infinite precision reals, we avoid the problem of floating point imprecision in the checker. 
Evaluating our checker suggests a trade-off between reliability and scalability. We hypothesize that the differences originate primarily due to the use of infinite-precision arithmetic and the choice of verification-oriented data structures. We leave a more detailed analysis for future work.

\subsection{Related Work} 
We hope to encourage further collaboration between verification and programming language communities.   
Notable is the recent success of industrial ITPs, such as Imandra~\cite{passmoreLessonsLearnedIndustrialization2021a} and   
F*~\cite{andriciSecuringVerifiedIO2024, merigouxCatalaProgrammingLanguage2021}.
Such languages are particularly suitable for AI verification, thanks to their automation and other modern features~\cite{desmartinCheckINNWideRange2022}.
 Our work opens the opportunity for future integration of DNN verifiers in these provers and ITPs more generally.  

\textbf{Imandra vs other Provers.} In principle, this work could be replicated in other proof assistants such as Isabelle/HOL, ACL2, LEAN, PVS or Rocq, and it could be interesting to do so.  Anecdotally, we find that Imandra's high level of proof automation, lemma discovery features, integrated bounded and unbounded verification, and efficient model execution make it an ideal environment for developing verified tools such as our verified proof checker.

\textbf{Proof-evidence production for SMT solvers}  is a known problem~\cite{Ne98, BaDeFo15,BaReKrLaNiNoOzPrViViZoTiBa22, DeBj11}, and we build on some experience in this domain. However, no SMT solver has been verified in an ITP as far as we are aware, the closest work in this direction comes from ITPs that integrate SMT solvers for proof automation and can verify the proof evidence~\cite{bohmeFastLCFStyleProof2010,EkMeTiKeKaReBa17}.

\textbf{Existing formalisations of the Farkas Lemma.} Although our proof of the Farkas lemma takes into consideration the previous experience in other ITPs~\cite{bessonModularSMTProofs2011,botteschFarkasLemmaMotzkin2019,passmoreACL2ProofsNonlinear2023}, 
it had to include substantial modifications  to cover the DNN case and in particular Imandra Theorem~\ref{ithm:reductiontogfl} is original relative to the cited papers.  Moreover, 
we prove Farkas lemma in a polynomial form (in Section~\ref{sec:check}). 
As opposed to the original matrix-based Farkas lemma~\cite{Va96}, the polynomial version yields easier automation in Imandra. 

\subsection{Future Work}
There are a number of directions in which we plan to extend this work.

\textbf{From Farkas vectors to specifications.} We plan to lift the soundness of the checker (Theorem~\ref{ilma:checker_soundness}) to the level of DNN verification queries, and, via a certified compilation procedure, to higher-order specification languages such as Vehicle~\cite{daggittVehicleHighLevelLanguage}. This would extend formal verification to encompass the left-hand-side blocks of Figure~\ref{fig:trust}.  Thus, this paper can be seen as a first step towards developing methods of \emph{proof-carrying neuro-symbolic code}~\cite{komendantskaya2025proofcarryingneurosymboliccode}.

\textbf{Optimisation of Proof Checking.} One of Marabou's key techniques for scaling to larger verification problems is the use of \emph{theory lemmas} for dynamic bound tightening. These lemmas characterize the connections between variables that participate in a \relu constraint, and are not proven directly by using the Farkas lemma. This feature is supported by the proof production in \cite{IsBaZhKa22}. Although it is possible to run DNN verification queries -- and certify their proofs -- without support of these theory lemmas, they are key to scaling to larger verification queries. We implemented them in Imandra, but certifying their soundness would be the next logical step to fully support proofs generated by Marabou.
While supporting DNNs with other piecewise-linear activations, such as \emph{maxpool} and \emph{sign}, is relatively striaghtforward, non-linear activations are more challenging: their verification relies on over-approximations, which Marabou does not currently generate proofs for. This is also grounded in theoretical results~\cite{IsZoBaKa23}. 

\textbf{Verification of cyber-physical systems with DNN components} is another possible application for presented results. For this,  DNN verifiers need to interface with languages that can express the system dynamics and/or probabilistic safety properties~\cite{daggittVehicleHighLevelLanguage}.
 The presented proof production will ensure that any such integration is sound.


\bibliographystyle{abbrv}
\bibliography{main}

\begin{thebibliography}{10}

\bibitem{affeldt2022mathcomp}
R.~Affeldt, Y.~Bertot, C.~Cohen, M.~Kerjean, A.~Mahboubi, D.~Rouhling, P.~Roux, K.~Sakaguchi, Z.~Stone, P.-Y. Strub, et~al.
\newblock {MathComp-Analysis: Mathematical Components Compliant Analysis Library}.
\newblock \url{https://math-comp.github.io/}, 2022.

\bibitem{AllamigeonK19}
X.~Allamigeon and R.~D. Katz.
\newblock {A Formalization of Convex Polyhedra Based on the Simplex Method}.
\newblock {\em Journal of Automated Reasoning}, 63(2):323--345, 2019.

\bibitem{AmCoYeMaHaFaKa23}
G.~Amir, D.~Corsi, R.~Yerushalmi, L.~Marzari, D.~Harel, A.~Farinelli, and G.~Katz.
\newblock {Verifying Learning-Based Robotic Navigation Systems}.
\newblock In {\em Proc. 29th Int. Conf. on Tools and Algorithms for the Construction and Analysis of Systems (TACAS)}, pages 607--627, 2023.

\bibitem{andriciSecuringVerifiedIO2024}
C.~Andrici, S.~Ciobaca, C.~Hritcu, G.~Mart{\'{\i}}nez, E.~Rivas, {\'{E}}.~Tanter, and T.~Winterhalter.
\newblock Securing verified {IO} programs against unverified code in {F}*.
\newblock {\em Proc. 51st Symposium on Principles of Programming Languages (POPL)}, pages 2226--2259, 2024.

\bibitem{BaReKrLaNiNoOzPrViViZoTiBa22}
H.~Barbosa, A.~Reynolds, G.~Kremer, H.~Lachnitt, A.~Niemetz, A.~N{\"o}tzli, A.~Ozdemir, M.~Preiner, A.~Viswanathan, S.~Viteri, Y.~Zohar, C.~Tinelli, and C.~Barrett.
\newblock {Flexible Proof Production in an Industrial-Strength SMT Solver}.
\newblock In {\em Proc. 11th Int. Joint Conference on Automated Reasoning (IJCAR)}, pages 15--35, 2022.

\bibitem{BaDeFo15}
C.~Barrett, L.~de~Moura, and P.~Fontaine.
\newblock {Proofs in Satisfiability Modulo Theories}.
\newblock {\em All about Proofs, Proofs for All}, 55(1):23--44, 2015.

\bibitem{BaIoLaVyNoCr16}
O.~Bastani, Y.~Ioannou, L.~Lampropoulos, D.~Vytiniotis, A.~Nori, and A.~Criminisi.
\newblock {Measuring Neural Net Robustness with Constraints}.
\newblock In {\em Proc. 30th Conf. on Neural Information Processing Systems (NeurIPS)}, 2016.

\bibitem{bessonModularSMTProofs2011}
F.~Besson, P.-E. Cornilleau, and D.~Pichardie.
\newblock {Modular {{SMT}} Proofs for Fast Reflexive Checking inside {{Coq}}}.
\newblock In {\em Proc. 38th Symposium on Principles of Programming Languages (POPL)}, pages 151--166, 2011.

\bibitem{bohmeFastLCFStyleProof2010}
S.~B{\"o}hme and T.~Weber.
\newblock Fast {{LCF-Style Proof Reconstruction}} for {{Z3}}.
\newblock In {\em Proc. 1st Int. Conf. on Interactive {{Theorem Proving}} (ITP)}, pages 179--194, 2010.

\bibitem{BoDeDwFiFlGoJaMoMuZhZhZhZi16}
M.~Bojarski, D.~Del~Testa, D.~Dworakowski, B.~Firner, B.~Flepp, P.~Goyal, L.~Jackel, M.~Monfort, U.~Muller, J.~Zhang, X.~Zhang, J.~Zhao, and K.~Zieba.
\newblock {End to End Learning for Self-Driving Cars}, 2016.
\newblock Technical Report. \url{http://arxiv.org/abs/1604.07316}.

\bibitem{botteschFarkasLemmaMotzkin2019}
R.~Bottesch, M.~W. Haslbeck, and R.~Thiemann.
\newblock {Farkas' Lemma and Motzkin's Transposition Theorem}.
\newblock {\em Archive of Formal Proofs}, 2019.

\bibitem{boyerComputationalLogic1979}
R.~S. Boyer and J.~S. Moore.
\newblock {\em {A Computational Logic}}.
\newblock {ACM Monograph Series. Academic Press, New York}, 1979.

\bibitem{BrBaJoWu24}
C.~Brix, S.~Bak, T.~T. Johnson, and H.~Wu.
\newblock {The Fifth International Verification of Neural Networks Competition (VNN-COMP 2024): Summary and Results}, 2024.
\newblock Technical report. \url{http://arxiv.org/abs/2412.19985}.

\bibitem{BrMuBaJoLi23}
C.~Brix, M.~N. M{\"u}ller, S.~Bak, T.~T. Johnson, and C.~Liu.
\newblock {First Three Years of the International Verification of Neural Networks Competition (VNN-COMP)}, 2023.
\newblock Technical report. \url{http://arxiv.org/abs/2301.05815}.

\bibitem{PLNNV25}
L.~Cordeiro, M.~Daggitt, J.~Girard, O.~Isac, T.~Johnson, G.~Katz, E.~Komendantskaya, E.~Manino, A.~Sinkarovs, and H.~Wu.
\newblock {Neural Network Verification is a Programming Language Challenge}.
\newblock In {\em Proc. 34th European Symposium on Programming (ESOP)}, {2025}.

\bibitem{daggittVehicleHighLevelLanguage}
M.~L. Daggitt, W.~Kokke, R.~Atkey, E.~Komendantskaya, N.~Slusarz, and L.~Arnaboldi.
\newblock {Vehicle: Bridging the Embedding Gap in the Verification of Neuro-Symbolic Programs}.
\newblock In {\em Proc. 10th Int. Conf. on Formal Structures for Computation and Deductionn, (FSCD)}, 2025.

\bibitem{Da63}
G.~Dantzig.
\newblock {\em {Linear Programming and Extensions}}.
\newblock Princeton University Press, 1963.

\bibitem{DeBj11}
L.~de~Moura and N.~Bj\o{}rner.
\newblock {Satisfiability Modulo Theories: Introduction and Applications}.
\newblock {\em Communications of the ACM}, 54(9):69--77, 2011.

\bibitem{desmartinNeuralNetworksImandra2022}
R.~Desmartin, G.~O. Passmore, and E.~Komendantskaya.
\newblock {Neural Networks in Imandra: Matrix Representation as a Verification Choice}.
\newblock In {\em Proc. 5th Int. Workshop of Software Verification and Formal Methods for ML-Enabled Autonomous Systems (FoMLAS) and 15th Int. Workshop on Numerical Software Verification (NSV)}, pages 78--95, 2022.

\bibitem{desmartinCheckINNWideRange2022}
R.~Desmartin, G.~O. Passmore, E.~Komendantskaya, and M.~Daggit.
\newblock {{CheckINN}}: {{Wide Range Neural Network Verification}} in {{Imandra}}.
\newblock In {\em Proc. 24th Int. Symposium on {{Principles}} and {{Practice}} of {{Declarative Programming}} (PPDP)}, pages 3:1--3:14, 2022.

\bibitem{DuDe06}
B.~Dutertre and L.~de~Moura.
\newblock {A Fast Linear-Arithmetic Solver for DPLL(T)}.
\newblock In {\em Proc. 18th Int. Conf. on Computer Aided Verification (CAV)}, pages 81--94, 2006.

\bibitem{Eh17}
R.~Ehlers.
\newblock {Formal Verification of Piece-Wise Linear Feed-Forward Neural Networks}.
\newblock In {\em Proc. 15th Int. Symp. on Automated Technology for Verification and Analysis (ATVA)}, pages 269--286, 2017.

\bibitem{EkMeTiKeKaReBa17}
B.~Ekici, A.~Mebsout, C.~Tinelli, C.~Keller, G.~Katz, A.~Reynolds, and C.~Barrett.
\newblock {SMTCoq: A Plug-in for Integrating SMT Solvers into Coq}.
\newblock In {\em Proc. 29th Int. Conf. Computer Aided Verification (CAV 2017)}, pages 126--133, 2017.

\bibitem{ElElIsDuMeGaPoGuBoCoKa24}
Y.~Elboher, R.~Elsaleh, O.~Isac, M.~Ducoffe, A.~Galametz, G.~Pov{\'e}da, R.~Boumazouza, N.~Cohen, and G.~Katz.
\newblock {Robustness Assessment of a Runway Object Classifier for Safe Aircraft Taxiing}.
\newblock In {\em Proc. 43rd Int Digital Avionics Systems Conf. (DASC)}, pages 1--6, 2024.

\bibitem{GeMiDrTsCgVe18}
T.~Gehr, M.~Mirman, D.~Drachsler-Cohen, P.~Tsankov, S.~Chaudhuri, and M.~Vechev.
\newblock {AI2: Safety and Robustness Certification of Neural Networks with Abstract Interpretation}.
\newblock In {\em Proc. 39th IEEE Symposium on Security and Privacy (SP)}, pages 3--18, 2018.

\bibitem{GoLuMaPuXiYu23}
D.~Gopinath, L.~Lungeanu, R.~Mangal, C.~Pasareanu, S.~Xie, and H.~Yu.
\newblock {Feature-Guided Analysis of Neural Networks}.
\newblock In {\em Proc. 26th Int. Conf. on Fundamental Approaches to Software Engineering (FASE)}, pages 133--142, 2023.

\bibitem{IsBaZhKa22}
O.~Isac, C.~Barrett, M.~Zhang, and G.~Katz.
\newblock {Neural Network Verification with Proof Production}.
\newblock In {\em Proc. 22nd Int. Conf. on Formal Methods in Computer-Aided Design (FMCAD)}, pages 38--48, 2022.

\bibitem{IsZoBaKa23}
O.~Isac, Y.~Zohar, C.~Barrett, and G.~Katz.
\newblock {DNN Verification, Reachability, and the Exponential Function Problem}.
\newblock In {\em Proc. 34th Int. Conf. on Concurrency Theory (CONCUR)}, 2023.

\bibitem{JiRi21}
K.~Jia and M.~Rinard.
\newblock {Exploiting Verified Neural Networks via Floating Point Numerical Error}.
\newblock In {\em Proc. 28th Int. Static Analysis Symposium (SAS)}, pages 191--205, 2021.

\bibitem{KaBaDiJuKo21}
G.~Katz, C.~Barrett, D.~Dill, K.~Julian, and M.~Kochenderfer.
\newblock {Reluplex: a Calculus for Reasoning about Deep Neural Networks}.
\newblock {\em Formal Methods in System Design (FMSD)}, 2021.

\bibitem{komendantskaya2025proofcarryingneurosymboliccode}
E.~Komendantskaya.
\newblock {Proof-Carrying Neuro-Symbolic Code}.
\newblock In {\em Computability in Europe, CiE'25}, 2025.

\bibitem{merigouxCatalaProgrammingLanguage2021}
D.~Merigoux, N.~Chataing, and J.~Protzenko.
\newblock {Catala: A Programming Language for the Law}.
\newblock {\em Proceedings of the ACM on Programming Languages}, 5:77:1--77:29, 2021.

\bibitem{10.1145/3689624}
K.~S. Namjoshi and L.~D. Zuck.
\newblock {Program Correctness through Self-Certification}.
\newblock {\em Communications of the ACM}, page 74–84, 2025.

\bibitem{Necula97}
G.~Necula.
\newblock Proof-carrying code.
\newblock In {\em Proc. 24th Symposium on Principles of Programming Languages (POPL)}, pages 106--119, 1997.

\bibitem{Ne98}
G.~Necula.
\newblock {\em {Compiling with Proofs}}.
\newblock Carnegie Mellon University, 1998.

\bibitem{passmoreACL2ProofsNonlinear2023}
G.~Passmore.
\newblock {{ACL2 Proofs}} of {{Nonlinear Inequalities}} with {{Imandra}}.
\newblock {\em Electronic Proceedings in Theoretical Computer Science}, 393:151--160, 2023.

\bibitem{passmoreImandraAutomatedReasoning2020}
G.~Passmore, S.~Cruanes, D.~Ignatovich, D.~Aitken, M.~Bray, E.~Kagan, K.~Kanishev, E.~Maclean, and N.~Mometto.
\newblock {The Imandra Automated Reasoning System (System Description)}.
\newblock In {\em Proc. 10th Int. Joint Conf. Automated Reasoning (IJCAR)}, pages 464--471, 2020.

\bibitem{passmoreLessonsLearnedIndustrialization2021a}
G.~O. Passmore.
\newblock {Some Lessons Learned in the Industrialization of Formal Methods for Financial Algorithms}.
\newblock In {\em Proc. 24th Int. Symposium on Formal Methods (FM)}, pages 717--721, 2021.

\bibitem{vass_repo}
K.~Sakaguchi.
\newblock vass.
\newblock \url{https://github.com/pi8027/vass}, 2025.

\bibitem{SaLa21}
M.~S{\"a}lzer and M.~Lange.
\newblock {Reachability Is NP-Complete Even for the Simplest Neural Networks}.
\newblock In {\em Proc. 15th Int. Conf. on Reachability Problems (RP)}, pages 149--164, 2021.

\bibitem{suzukiOverviewDeepLearning2017}
K.~Suzuki.
\newblock {Overview of Deep Learning in Medical Imaging}.
\newblock {\em Radiological Physics and Technology}, 10(3):257--273, 2017.

\bibitem{szegedyIntriguingPropertiesNeural2013}
C.~Szegedy, W.~Zaremba, I.~Sutskever, J.~Bruna, D.~Erhan, I.~Goodfellow, and R.~Fergus.
\newblock {Intriguing Properties of Neural Networks}, 2013.
\newblock Technical report. \url{http://arxiv.org/abs/1312.6199}.

\bibitem{Va96}
R.~Vanderbei.
\newblock {Linear Programming: Foundations and Extensions}.
\newblock {\em Journal of the Operational Research Society}, 1996.

\bibitem{wangBetaCROWNEfficientBound2021}
S.~Wang, H.~Zhang, K.~Xu, X.~Lin, S.~Jana, C.-J. Hsieh, and J.~Z. Kolter.
\newblock {Beta-CROWN: Efficient Bound Propagation with Per-neuron Split Constraints for Neural Network Robustness Verification}.
\newblock {\em Advances in Neural Information Processing Systems}, 34:29909--29921, 2021.

\bibitem{WuIsZeTaDaKoReAmJuBaHuLaWuZhKoKaBa24}
H.~Wu, O.~Isac, A.~Zelji\'c, T.~Tagomori, M.~Daggitt, W.~Kokke, I.~Refaeli, G.~Amir, K.~Julian, S.~Bassan, P.~Huang, O.~Lahav, M.~Wu, M.~Zhang, E.~Komendantskaya, G.~Katz, and C.~Barrett.
\newblock {Marabou 2.0: A Versatile Formal Analyzer of Neural Networks}.
\newblock In {\em Proc. 36th Int. Conf. on Computer Aided Verification (CAV)}, 2024.

\bibitem{Zombori2021FoolingAC}
D.~Zombori, B.~B{\'a}nhelyi, T.~Csendes, I.~Megyeri, and M.~Jelasity.
\newblock {Fooling a Complete Neural Network Verifier}.
\newblock In {\em Proc. 9th Int. Conf. on Learning Representations (ICLR)}, 2021.

\end{thebibliography}

\end{document}